\documentclass[a4paper,11pt,twoside]{article}
\usepackage[T1]{fontenc}
\usepackage[utf8]{inputenc}
\usepackage{lmodern}
\usepackage{graphicx}
\usepackage[english]{babel}
\usepackage{bbm}

\usepackage{appendix}
\usepackage{soul}
\usepackage{fvextra}

\usepackage[bookmarks=false]{hyperref}
\hypersetup{colorlinks=true, linktocpage=true, pdfstartpage=3, pdfstartview=FitV,
    breaklinks=true, pageanchor=true,
    plainpages=false, bookmarksnumbered, bookmarksopen=true, bookmarksopenlevel=1,%
    hypertexnames=true, pdfhighlight=/O,
     urlcolor=blue!79!black, linkcolor=blue!79!black, citecolor=blue!79!black}

\usepackage[a4paper,tmargin=2.5truecm,bmargin=3truecm,outer=1.8truecm,inner=1.8truecm,twoside,verbose=false
]{geometry}

\newcommand{\textfrac}[2]{#1/#2}

\usepackage{amsmath}
\usepackage{cases}
\usepackage{amssymb}
\usepackage{amsthm}
\newtheoremstyle{mytheoremstyle} 
    {10pt}                    
    {8pt}                    
    {\itshape}                   
    {}                           
    {\scshape}                   
    {.}                          
    {.5em}                       
    {}  

\theoremstyle{mytheoremstyle}

\newtheorem{theorem}{Theorem}[section]
 \newtheorem{corollary}[theorem]{Corollary}
 \newtheorem{lemma}[theorem]{Lemma}

 \newtheorem{proposition}[theorem]{Proposition}
 \newtheoremstyle{definition} 
    {8pt}                    
    {5pt}                    
    {}                   
    {}                           
    {\scshape}                   
    {.}                          
    {.5em}                       
    {}  

\theoremstyle{definition}
\newtheorem{definition}[theorem]{Definition}

\newtheorem*{acknowledgements}{Acknowledgements}
\newtheorem{remark}[theorem]{Remark}

\def\[#1\]{%
\begin{align}#1\end{align}%
}
\newenvironment{proof*}[1]{\renewcommand{\proofname}{#1}%
\begin{proof}}
{\end{proof}}

\usepackage{mathtools}
\renewcommand{\mathbf}{\boldsymbol}
\usepackage{xcolor}
\usepackage{csquotes}

\usepackage[backend=bibtex, maxbibnames=99, natbib, backref=false, style=alphabetic,giveninits=true]{biblatex}
\usepackage{commath}
\newcommand{\tqs}{\mathrel{:}}

\newcommand{\bbC}{\mathbb{C}}
\newcommand{\C}{\bbC}

\newcommand{\bbR}{\mathbb{R}}
\newcommand{\R}{\bbR}

\newcommand{\bbZ}{\mathbb{Z}}
\newcommand{\Z}{\bbZ}

\newcommand{\kH}{\mathfrak{H}}

\newcommand{\rmI}{\mathrm I}

\newcommand{\defi}{\coloneqq}

\newcommand{\grp}[1]{\mathrm{#1}}
\newcommand{\grpO}{\grp{O}}
\newcommand{\modulus}[1]{\lvert #1\rvert}
\DeclareMathOperator{\liftarg}{\widetilde{\arg}}
\newcommand{\liftsqrt}[1]{\sqrt[\sim]{#1}}
\newcommand{\liftR}{\widetilde{R}}
\newcommand{\rmH}{\mathrm H}
\title{\textbf{Borel summability of the $\mathbf{1/N}$ expansion\\in quartic $\mathbf{\grp{O}(N)}$-vector models}}
 
\usepackage{authblk}
\author[1]{L. Ferdinand}
\author[2,3,4]{R. Gurau}
\author[2]{C.I. Perez-Sanchez}
\author[5]{ F. Vignes-Tourneret}
\affil[1]{\normalsize \itshape Laboratoire de Physique des 2 infinis Irène Joliot-Curie - IJCLab - \linebreak  UMR 9012, CNRS, Universit\'e Paris-Saclay \linebreak Bât. 210, 91405 Orsay cedex, France 	\authorcr \hfill}
\affil[2]{\normalsize\itshape 
	Institut f{\"u}r Theoretische Physik der Universität Heidelberg \linebreak Philosophenweg 19, 69120 Heidelberg, Germany
	\authorcr \hfill}
\affil[3]{\normalsize \itshape 
	CPHT, CNRS, Ecole Polytechnique, Institut Polytechnique de Paris\linebreak  Route de Saclay,  91128 PALAISEAU, 
	France
	\authorcr \hfill}
\affil[4]{\normalsize\itshape 
	Perimeter Institute for Theoretical Physics\linebreak  31 Caroline St. N, N2L 2Y5, Waterloo, ON,
	Canada	\authorcr \hfill}
	
\affil[5]{\normalsize\itshape Univ Lyon, CNRS, Université Claude Bernard Lyon 1\linebreak
UMR 5208, Institut Camille Jordan\linebreak
    F-69622 Villeurbanne, France
	\authorcr \hfill
	\authorcr
	Emails: \texttt{leonard.ferdinand@ijclab.in2p3.fr,
	gurau@thphys.uni-heidelberg.de, perez@thphys.uni-heidelberg.de, 
	vignes@math.univ-lyon1.fr}
      \authorcr \hfill}

\begin{document}
\maketitle

\begin{abstract}
    We consider a quartic $\grpO(N)$-vector model. Using the Loop
    Vertex Expansion, we prove the Borel summability in $1/N$ along
    the real axis of the partition function and of the connected
    correlations of the model. The Borel summability holds uniformly
    in the coupling constant, as long as the latter belongs to a cardioid like domain of the complex plane, avoiding the negative real axis.
\end{abstract}

\section{Introduction}

The Loop Vertex Expansion (LVE) was introduced by Rivasseau in 2007 \cite{Rivasseau2007aa} as a new tool in constructive field theory in order to deal with matrix fields. It was then successfully applied to general tensor fields \cite{Gurau2013ac,Delepouve2014aa,Rivasseau2016aa}.
For a general exposition in zero dimensions, close to the topic of this
article, see \cite{Rivasseau2009aa}.  The outcome of the LVE is
an expression of the free energy, as well as the generating function
of connected moments (or cumulants) as a sum over trees instead of
connected graphs.
As the number of trees increases only exponentially with the number of vertices and the contribution
of each tree is exponentially bounded, the resulting series is convergent. The two
ingredients of this expansion are the Hubbard–Stratonovich
 \cite{Hubbard,Stratonovich} intermediate field
representation and the Brydges-Kennedy--Abdesselam-Rivasseau (BKAR) formula
\cite{Brydges1987aa,Abdesselam1995aa}. 

In this paper we study the Borel summability in $1/N$ of the free
energy and the cumulants of the quartic $\grpO(N)$-vector model in
zero dimensions using the LVE (see Section \ref{sec-model} for the definition of the model). Note that here we are not interested
in the pertubative expansion (the expansion at small coupling
constant), which is well-understood for the quartic $
\grpO(N)$-vector model and is Borel summable in $0$ dimensions \cite{Rivasseau2007aa} and in $2$ dimensions \cite{Eckmann1974aa}. On the contrary, Borel summability in $1/N$ is less explored. The associated two-dimensional Euclidean quantum field theory was studied in \cite{BillionnetRenouard}, where the authors prove the 
 Borel summability of the partition function and of the moments of the
$\frac{g}{N}\norm{\phi}^{4}_2$ measure. But they discuss neither the free energy
nor the cumulants. Passing between the two is rather non trivial as one 
needs to take a logarithm.  The raison d'\^etre of the LVE is to take this
logarithm rigorously and uniformly in $N$. A related model, the spherical $\grpO(N)$ model (or
non-linear $\sigma$-model), has been studied in \cite{Kupiainen1980aa, Frohlich1982ld} where the authors 
showed that the partition function and the correlation functions 
at high enough temperature are Borel summable in $1/N$.
However, contrary to the model we study here, the spherical $\grpO(N)$ model does not have any 
issues of convergence at large field as the field is restricted to belong to 
the sphere $\mathbb S^{N-1}$. 

Techniques similar to the ones we use in this paper have been introduced 
in \cite{Gurau2014aa} for $N\times N$ matrices. However, only the Borel summability of the
perturbative expansion in the coupling constant has been established 
in \cite{Gurau2014aa}: the status of the  $1/N$ series has not been analyzed. The generalization of Borel summability results in $1/N$ to the case of matrices is not 
straightforward: contrary to the vector case, we do not have a representation of the partition function 
(with sources) in which $N$ is just a parameter. Consequently it has been impossible so far to prove that such functions can be extended to analytic functions in $1/N$ in some domain.\\

In this article  the free energy and the generating function of cumulants
of the quartic $\grpO(N)$ vector model are considered 
as functions of the coupling constant $g$ and of $1/N$. We 
look for the largest domain in the
$(g,1/N)$-plane allowing their bivariate analytic continuation. 
After introducing the model in Section \ref{sec-model}, we present both the main tools
and the two main results in Section \ref{sec:main}, namely the analyticity (Thm. \ref{THM1}) and the Borel summability (Thm. \ref{THM2})  domains of the free energy and the cumulants. We obtain that if $\lvert\arg g+\arg 1/N\rvert<3\pi/2$, the free energy and the cumulants are analytic in a cardioid shaped domain in $g$ and, for $\modulus{\arg g}<{\pi}$ they are Borel summable in $1/N$ along
the real axis uniformly in $g$ for $g$ in a slightly smaller cardioid domain. The proofs of these theorems are presented in Sections \ref{sec4} and  \ref{sec5},
respectively. In order to keep this
article self-contained, we recall 
BKAR formula in Lemma \ref{thm:BKAR}, but other useful tools also appear in the appendix.

\footnotesize
\begin{acknowledgements}
L.F. is supported by the EDPIF. R. G. and C.I.P-S. have been supported by the European Research Council (ERC) under the European Union’s Horizon 2020 research and innovation program (grant agreement No818066) and by the Deutsche Forschungsgemeinschaft (DFG, German Research Foundation) under Germany’s Excellence Strategy EXC-2181/1 -
390900948 (the Heidelberg STRUCTURES Cluster of Excellence). 
\end{acknowledgements}
\normalsize

\section{The model and the partition function}
\label{sec-model}
Before introducing the model, let us adopt the following notation:
\begin{itemize}
    \item We denote by $I_n$ the identity matrix on $\R^n$ and by $\mathbbm{1}_n$ the $n\times n$ matrix with all entries equal to $1$.
    \item Let $C \in M_n(\mathbb{R})$ be symmetric positive and $X,Y
      \in \R^n$. We write $\langle X,Y \rangle_C$ for $\sum_{1\leq i,j
        \leq n} X_i C_{ij} Y_j$. If $C=I_n$, we omit it. We
      denote $\langle X,X \rangle$ by $\norm{X}^2$. Whenever $X
      \in \mathbb{R}^n$ is the argument of a function $F :
      \mathbb{R}^n \rightarrow \mathbb{C}$, we write 
      \[\nonumber
      \langle
      \partial,\partial \rangle_C F(X)=
         \sum_{1 \leq i,j \leq n} C_{ij} 
      \frac{\partial^2F}{\partial{x_i}\partial{x_j}}(X) \; .
    \]
    \item Let $C \in M_n(\R)$ be symmetric positive semi-definite. We denote by
      $\mu_C$ the centered Gaussian probability distribution of
      covariance $C$ on $\R^n$. Note that it exists and is unique 
      (see appendix~\ref{gaussexp}) even if $C$ is degenerate.
    \item If $F:\mathbb{R}^n \rightarrow \mathbb{C}$, we denote $\mathbb{E}_C[F(X)]$ the expectation of $F$ with respect to $\mu_C$:
    \begin{equation} \nonumber
        \mathbb{E}_C[F(X)]=\int d\mu_C(X) F(X)=[e^{\frac{1}{2}\langle \partial,\partial \rangle_C}F(X)]_{X=0} \;.
    \end{equation}
    \item We write $a \lesssim b$ if there is a constant $K >0$ such
      that $a \leq K b$. If we want to specify that $K$ depends on some parameter 
      $\alpha$, we write $a \lesssim_\alpha b$.
    \item Throughout this paper, we denote $1/N$ by $\epsilon$ when promoted to a complex variable.
   \end{itemize}

Let $N$ be a positive integer and $g\in \{z\in\C\mid \Re z>0\}$. The
zero-dimensional quartic $\grpO(N)$-vector model is a probability distribution
$\nu$ on $\R^N$ defined as a perturbed Gaussian distribution in the following way: denoting by $\mathbb{E}$ the expectation with respect
to $\nu$, for all $F:\R^N\rightarrow\C$ $\nu$-mesurable, the
expectation of $F$ is
\begin{align}\nonumber
    \mathbb{E}[F(X)]=\frac{\mathbb{E}_{I_N}[e^{-\frac{g}{8N}\norm{X}^4}F(X)]}{\mathbb{E}_{I_N}[e^{-\frac{g}{8N}\norm{X}^4}]} \;.
\end{align}
The Fourier-Laplace transform of the measure, also known in the physics literature as the 
\textit{partition function with sources} $J\in \R^N$, denoted $Z(g,1/N;J)$, is:
\begin{align}  \nonumber
     Z\big(g,\frac1N;J\big) = \mathbb{E}_{I_N}[e^{-\frac{g}{8N}\norm{X}^4}]\,\mathbb{E}[e^{\sqrt{N} \langle J,X\rangle}]=\mathbb{E}_{I_N}[e^{-\frac{g}{8N}\norm{X}^4+\sqrt{N} \langle J,X\rangle}] \;.
\end{align}
In particular, the \textit{partition function} of $\nu$, $Z(g,1/N;0) =  \mathbb{E}_{I_N}[e^{-\frac{g}{8N}\norm{X}^4}]$, is the normalisation constant of that measure.
\begin{remark}
Note that we have made a particular choice of scaling of the
sources $J$ with $N$. This scaling ensures that all the cumulants (see below for details) are non-trivial
in the large $N$ limit. In the absence of this scaling, in the large $N$ limit
the $\grpO(N)$-vector model is a Gaussian model with a
complicated covariance corresponding to the resummation of the dominant diagrams in the large $N$ limit, the so-called cactus diagrams. 
\end{remark}

From now on, we will switch to an integral notation, more adapted to
the LVE, and more reminiscent of the functional integration in quantum
field theory. In particular, in accordance to the usual notation of
quantum field theory, we denote $\phi \in \R^N$ the random vector so
that the partition function with sources rewrites:
\begin{align}\label{ON_in0dim}
    Z\big(g,\frac1N;J\big)&=\int_{\R^N}\frac{d^N\phi}{(2\pi)^{N/2}}\; e^{-\frac
      12\norm\phi^2-\frac g{8N}\norm\phi^4+ \sqrt{N} \langle
      J,\phi\rangle}=\int d\mu_{I_N}(\phi)\; e^{-\frac
      g{8N}\norm\phi^4+ \sqrt{N} \langle J,\phi\rangle} \;.
\end{align}

Our aim is to study the expansion in $1/N$ of
the partition function and the cumulants of the measure $\nu$. At fixed
$N\in\Z_{>0}$, the integral in eq.~\eqref{ON_in0dim} is absolutely
convergent iff $\Re g\geq 0$ and defines $ Z(g,\tfrac{1}{N} ;J)$ as a holomorphic
function of $g$ for all $g \in \{ z\in \C \mid
\Re z>0\}$. 

In \cite{Rivasseau2007aa} it was noted that performing a change of variables 
(known as the Hubbard-Stratonovich transformation, or intermediate field
representation) one can obtain a convergent expansion for the logarithm of the partition function. We thus insert in eq.~\eqref{ON_in0dim} the Hubbard-Stratonovich intermediate field representation ($\imath = \sqrt{-1}$):
\begin{equation} \nonumber
 e^{-\frac{x^2}{2}}={\frac{1}{\sqrt{2 \pi }}} \int_{\mathbb{R}} dy \; e^{
  - \frac{y^2}{2} + \imath x y }   = \int d\mu_{1}(y) \; e^{\imath xy} \;,
\end{equation}
for the quartic interaction term, $x=\frac{1}{2}\sqrt{g/N}\norm\phi^2$ and obtain:
\begin{align*}
 Z\big(g,\frac 1N ;J\big) & =  \int d\mu_{{I}_N}(\phi)\int d\mu_1(\sigma) \;
 e^{\imath \frac{1}{2}\sqrt{g/N}\, \norm\phi^2 \sigma
   +\sqrt{N} \langle J, \phi \rangle} \crcr &= \int
 d\mu_{1}(\sigma) \; e^{\frac{N}{2}\ln R(\sigma, g/N) +
   \frac{N}{2}R(\sigma,g/N)\,\norm{J}^2}\nonumber \\&=\int
 d\mu_{\frac1N}(\sigma) \; e^{\frac{N}{2}\ln R(\sigma, g) +
   \frac{N}{2}R(\sigma,g)\,\norm{J}^2}
\end{align*}
where $R:\C^2 \setminus \{ (\frac{ 1}{\imath\sqrt{z}},z) \tqs z \in
\C^* \} \rightarrow \C$, $R(\sigma,z) =  {(1-\imath\sqrt{z}\sigma)}^{-1}$ is called the resolvent.
\begin{remark}
This transformation renders the O($N$) invariance explicit: the partition function depends only on the norm of the sources.    
\end{remark}

We observe that, at fixed $N\in\Z_{>0}$, $Z$ can
be analytically continued in $g$ to all $\C \setminus
\R_-$. The intermediate field representation also makes
$1/N$ an explicit parameter \cite{Kupiainen1080aa} in an integral representation of $ Z(g,1/N;J)$, contrary to eq.~\eqref{ON_in0dim}  where $N$ is also present implicitly in the dimension of the integral. One can then study the analyticity properties of $Z(g,1/ N;J)$ seen as a function of the variable $\epsilon=1/N$ that, since we are interested in Borel summability along the positive real axis, we promote to $\rmH=\{z\in\C\mid \Re z>0\}$. We parameterize $\C^*$ as $\{(\modulus{z},\alpha)\in\R_+^*\times (-\pi,\pi]\}$ and for $z=(\modulus{z},\alpha)\in\C^*$, we write $\arg z=\alpha$, and we use the same parametrization for $\rmH=\{z\in\C^*\mid \modulus{\arg z}<\pi/2\}$. Regarding $g$, since it appears in a square root in the resolvent, it is wiser to take it to be an element of the Riemann surface of the square root, whose basic properties are recall hereafter:
\begin{definition}[Riemann surface of the square root]
Let us denote by $\sqrt{\phantom{z}}$ the principal branch of the complex square root defined on $\C\setminus\R_-$. Let $\Sigma$ be the associated Riemann surface. We write $\liftsqrt{\phantom{\omega}}$ for the analytic continuation of $\sqrt{\phantom{z}}$ to $\Sigma$. $\Sigma$ is a $2$-sheeted covering of $\C^*$ and can be parameterized as $\{(\modulus{z},\alpha)\in\R_+^*\times (-2\pi,2\pi]\}$ and for $z=(\modulus{z},\alpha)\in\Sigma$, we write $\liftarg z=\alpha$. Let $\rmI:z=(\modulus{z},\alpha)\in \C\setminus\R_-\mapsto (\modulus{z},\alpha)\in\Sigma$ be the canonical injection of $\C\setminus\R_-$ into $\Sigma$ and $\Pi$ be the projection of $\Sigma$ onto $\C^*$ namely $\Pi(\modulus{z},\alpha)=\modulus{z} e^{\imath\alpha}$. The two sheets of $\Sigma$ correspond to the two possible determinations of the square root on $\C\setminus\R_-$ : for $z\in\Sigma$ such that $\liftarg z\in(-\pi,\pi)$, $\liftsqrt z=\sqrt{\Pi z}$ and if $\liftarg z\in(-2\pi,-\pi)\cup(\pi,2\pi]$, $\liftsqrt z=-\sqrt{\Pi z}$ (but note that $(\liftsqrt{z})^2$ is always equal to $\Pi z $). We will denote by $\Sigma^+$ (resp.\@ $\Sigma^-$) the sheet of $\Sigma$ corresponding to $\sqrt{\phantom{z}}$ (resp.\@ to $-\sqrt{\phantom{z}}$), and we will also denote $\liftR(\sigma,z)\defi (1-\imath\sigma\liftsqrt{z})^{-1}$.
\end{definition}
Therefore, from now on, the coupling constant $g$ is an element of $\Sigma$, and we aim to find the maximal domain of analyticity of the free energy and the cumulants as functions of $(g,\epsilon)\in
\Sigma\times\rmH$. Since this will bring us to constantly deal with with the arguments of $g $ and $\epsilon$, in the rest of the article, we use the following convention:
\begin{equation*}
  \text{we will use indifferently $\liftarg g$ and $\varphi$, as well as $\arg \epsilon$ and $\theta$}.
    \end{equation*}
At this point, partition function rewrites as
\begin{align*}
 Z(g,\epsilon ;J) &= \int d\mu_{\epsilon}(\sigma)\; e^{\frac{1}{2\epsilon}\ln \liftR(\sigma, g) +
   \frac{1}{2\epsilon}\liftR(\sigma,g)\,\norm{J}^2}\,,
\end{align*}
which enables to analytically continue it from $\rmI(\R_+^*)\times\{1/N\mid N\in \Z_{>0}\}$ to $\rmI(\C\setminus\R_-)\times \rmH$. In order to extend this continuation to some wider subdomain of $\Sigma\times\rmH$, for all $\psi \in (\theta-\pi/2,\theta+\pi/2)$ we define $ Z_\psi(g,\epsilon ;J)$ by
\begin{align}
     Z_\psi(g,\epsilon ;J) &= \int_{e^{\imath \frac{ \psi}{2}
       }\R} \frac{d\sigma}{\sqrt{2\pi\epsilon}} \;
     e^{-\frac{\sigma^2}{2\epsilon}+\frac{1}{2\epsilon}\ln \liftR(\sigma, 
       g) + \frac{1}{2\epsilon}\liftR(\sigma, g)\,\norm{J}^2}\nonumber\\
     &=
     \int_{\R} \frac{d\sigma}{e^{-\imath \frac\psi2}\sqrt{2\pi \epsilon }} \;
     e^{-\frac{\sigma^2}{2\epsilon e^{-\imath
           \psi}}+\frac{1}{2\epsilon}\ln \liftR(\sigma e^{\imath \frac\psi2}, g)
       + \frac{1}{2\epsilon}\liftR(\sigma e^{\imath \frac\psi2}, g)\,\norm{J}^2}
     \nonumber
     \\ &= \int d\mu_{\epsilon e^{-\imath \psi}}(\sigma) \;
     e^{\frac{1}{2\epsilon}\ln \liftR(\sigma e^{\imath \frac\psi2} , g) +
       \frac{1}{2\epsilon}\liftR(\sigma e^{\imath \frac\psi2}, g)\,\norm{J}^2} \;. \label{eq-Ztiltcontour}
\end{align}
The integral is convergent and, furthermore, $Z_\psi(g,\epsilon ;J)=Z(g,\epsilon ;J)$ is independent of $\psi$. Indeed, let $s_{g,\epsilon,J}:\sigma\in\C\mapsto \frac{1}{\sqrt{2\pi \epsilon }} \;
     e^{-\frac{\sigma^2}{2\epsilon}+\frac{1}{2\epsilon}\ln \liftR(\sigma , g)
       + \frac{1}{2\epsilon}\liftR(\sigma , g)\,\norm{J}^2} $ so that $ Z_\psi(g,\epsilon ;J)=\int_{\R}e^{\imath \frac\psi2} s_{g,\epsilon,J}(\sigma e^{\imath \frac\psi2})d\sigma$. We then have $\frac{d}{d\psi}Z_\psi(g,\epsilon ;J)=\int_{\R}\frac{\imath}{2}e^{\imath \frac\psi2}d\sigma [s_{g,\epsilon,J}(\sigma e^{\imath \frac\psi2})+\sigma {s'}_{g,\epsilon,J}(\sigma e^{\imath \frac\psi2})]=0$ by integration by part.\\

Before going to the analyticity domain of the free energy and the cumulants, for the sake of comparison, we note the following result:
\begin{proposition}\label{prop1}
The partition function with sources of the zero-dimensional $
\grpO(N)$-vector model, $ Z(g,\epsilon ;J)$, can be analytically
continued in $(g,\epsilon)$ from $\rmI(\R_+^*)\times \{1/N \mid N \in \Z_{>0}\}$ to the following domain of $\Sigma\times\rmH$:
\begin{equation*}
   \mathfrak{B} = \Big\{(g,\epsilon) \in \Sigma\times\rmH \,\big\vert\, \liftarg g+\arg \epsilon \in \big(-\frac{3\pi}{2},\frac{3\pi}{2}\big) \Big\}.
\end{equation*}
\end{proposition}

\begin{remark}
In the sequel, to prove Borel summability of the free energy or the cumulants of $\nu$, we will rely on the Nevanlinna-Sokal theorem \cite{Sokal1980aa}. One important hypothesis of this theorem is analyticity in a disk tangent to the imaginary axis and centered at a positive real number. We call such a domain a Sokal disk, see remark~\ref{remarksokaldisk}. Note that for $g \in \mathrm I(\C\setminus\R_-)$, the analyticity domain in the $\epsilon$-plane of the partition function with sources contains indeed a Sokal disk as for all $\theta \in
(-\pi/2,\pi/2)$ and $\varphi\in(-\pi,\pi)$, $\varphi+\theta\in(-3\pi/2,3\pi/2)$.
\end{remark}
In order to prove Proposition \ref{prop1}, we need the following bound on the resolvent:
\begin{lemma}
For all $(\sigma, g) \in \mathbb{C}^*\times\Sigma $,
\begin{equation}
\modulus{\liftR(\sigma,g)} \leq \frac{1}{\modulus{\cos({\arg \sigma+ \frac{1}{2}\liftarg{g}})}}\,.
\label{resbound}
\end{equation}
\end{lemma}
\noindent
This bound is trivial for $\liftsqrt{g}\sigma\in \imath\R$, which reflects the fact that the resolvent has a pole at $\sigma=1/\imath\liftsqrt{g}$.
\begin{proof}
Directly stems from $\modulus{1-\imath z}\geq\modulus{\cos\arg z}$.
\end{proof}

\begin{proof}[Proof of Proposition \ref{prop1}]
We start with the intermediate field representation \eqref{eq-Ztiltcontour}. Let $\kH_\psi$ be the following manifold:
\begin{equation} \nonumber
    \kH_\psi\defi\big\{(\sigma,g,\epsilon)\in\C\times\Sigma\times\rmH\mid\sigma e^{\imath \frac\psi2} \liftsqrt g \in\C\setminus\imath\R\big\}.
\end{equation}
We let $f_\psi$ from $\kH_\psi$ to $\C$ be the integrand in eq.~\eqref{eq-Ztiltcontour}:
\begin{equation}\nonumber
    f_\psi(\sigma,g,\epsilon)=\frac{1}{e^{-\imath \frac\psi2} \sqrt{2\pi \epsilon}}e^{-\frac{\sigma^2}{2\epsilon e^{-\imath \psi}}+\frac{1}{2\epsilon}\ln\liftR(\sigma e^{\imath \frac\psi2}, g) + \frac{1}{2\epsilon}\liftR(\sigma e^{\imath \frac\psi2},g)\,\norm{J}^2}.
\end{equation}
$f_\psi$ is holomorphic on $\kH_\psi$ and $\int_\R f_\psi(\sigma,g,\epsilon)d\sigma$ coincides with \eqref{ON_in0dim} for $(g,\epsilon)\in\rmI(\R_+^*)\times\{1/N\mid N\in\Z_{>0}\}$. For all $\sigma\in\R$, $(g,\epsilon)\mapsto f_\psi(\sigma,g,\epsilon)$ is holomorphic on $ \mathfrak{A}_\psi\defi\{(g,\epsilon)\in\Sigma\times\rmH\mid e^{\imath \frac\psi2} \liftsqrt g\in\C\setminus\imath\R\}$, that has two connected components, namely 
\begin{equation}\label{Apsiplus}
\mathfrak{A}_{\psi}^+=\big\{(g,\epsilon)\in\Sigma\times\rmH\mid\varphi\in(-\pi-\psi,\pi-\psi)\big\}    
\end{equation}
and $ \mathfrak{A}_{\psi}^-=\{(g,\epsilon)\in\Sigma\times\rmH\mid\varphi\in (-2\pi,2\pi]\setminus(-\pi-\psi,\pi-\psi)\}$. Moreover as $\modulus{\ln\liftR(\sigma e^{\imath \frac\psi2}, g)} \le \modulus{\ln\modulus{\liftR(\sigma e^{\imath \frac\psi2}, g)}}+\pi$, thanks to the bound \eqref{resbound}, the integral of $f_\psi$ is absolutely convergent, uniformly in $g$ and $\epsilon$, on any compact of $\mathfrak{A}_\psi^+$. Thus, it defines an analytic continuation of $Z_\psi=Z$ to $\mathfrak{A}_\psi^+$. Therefore, $Z$ is analytic on $\bigcup_{\psi \in (\theta-\pi/2,\theta+\pi/2)}  \mathfrak{A}_\psi^+=\mathfrak{B}$ which concludes the proof of
Proposition~\ref{prop1}.
\end{proof}

\begin{remark}
This analytic continuation is the largest one that can be found thanks to the tilt of the contour of integration, since for $|\psi-\theta|\geq\pi/2$, the integral \eqref{eq-Ztiltcontour} becomes divergent.
\end{remark}

\begin{remark}
In order to clarify why the tilting of the contour was needed,
consider the following. Suppose we are interested
in the function $h:\{\Re z>0\}\rightarrow\C, z\mapsto
\int_{\R_+}e^{-zt}dt$ and ignore that $f(z)=1/z$. Clearly, $h$ is
analytic on its domain of definition. We aim to analytically continue $h$ to 
some maximal domain. To this end, we observe that $h_{\psi}: z=\modulus{z}e^{\imath\alpha}\mapsto
\int_{e^{\imath\psi}\R_+}e^{-zt}dt$ is analytic iff
$\modulus{\alpha+\psi}<\pi/2$. Moreover, if
$\modulus{\psi}<\pi$, the domains of analyticity of $h$ and $h_{\psi}$
overlap, and $h=h_{\psi}$ where they are both analytic. Thus $h_{\psi}$ is an
analytic continuation of $f$ to a Riemann sheet. One needs to check whether this analytic 
continuation has a discontinuity at the real negative
axis (in which case $0$ is a branch point) or a pole. In our case 
one gets a pole, but applying the same strategy to $z\mapsto
\int_{\R}e^{-zt^2}dt$ one obtains a branch point of order $2$. We apply the same strategy to $Z$.
\end{remark}

Our aim is to obtain similar results for the free energy
and the cumulants. As such quantities depend on the logarithm of the partition 
function $Z$ and $Z$ has zeroes, they will not simply inherit the analyticity properties of
$Z$: we expect that the domain of analyticity of $\ln Z$ is smaller than the one of $Z$. Since $Z(0,\epsilon;J)$ is non vanishing, for $g$ close to $0$, $Z(g,\epsilon;J)$ is non vanishing too. However, for $g$ real negative $Z$ is discontinuous and $g=0$ does not belong to the domain of analyticity of $Z$ or $\ln Z$. 
In order to identify some domain of analyticity of $\ln Z$
we will rewrite it as a uniformly convergent series of
analytic functions. This series is indexed by trees and converges for a small enough
coupling constant thereby defining $\ln Z$ in some domain. This is in contrast with the perturbative expansion which writes the partition function and the cumulants as divergent series.\\

The core of our arguments heavily relies on the \textit{Loop Vertex Expansion} (LVE) \cite{Magnen2008aa,Rivasseau2007aa}, which we now present. 

\section{The Loop Vertex Expansion of the cumulants}\label{sec:main}

In this section, we perform the LVE of the cumulants defined hereafter:
\begin{definition}[The rescaled cumulants] For all $k\geq 1$, one defines
  the \textit{rescaled cumulant }of order $2k$,
  $\mathfrak{K}^{2k}(g,\epsilon)$, by the following relation:
\begin{align} \label{scaled_cumul}
\left( \sum_{\pi\in
       P_2(2k)} \prod_{(a_i,a_j) \in \pi} \delta_{a_i,a_j} \right)
   \mathfrak{K}^{2k}_\psi(g,\epsilon)&:=  
   \epsilon \frac{\partial^{2k}}{\partial_{J_{a_1}}...\partial_{J_{a_{2k}}}} \left. \ln{Z_\psi(g,\epsilon ; J)} \right|_{J=0},
\end{align}
where $ P_2(2k)$ is the set of pairings of $2k$ elements and let $\mathfrak K^{2k}(g,\epsilon)=\mathfrak K^{2k}_{\psi=0}(g,\epsilon)$. 
\end{definition}
This scaling is chosen as to have a well-defined large-$N$ limit 
and the advantage of using $\mathfrak{K}^{2k}$ over the RHS of \eqref{scaled_cumul}
is that the former is manifestly O($N$) invariant. Since only $\mathfrak{K}^{2k}$ will 
appear below, we refer to them as the cumulants.\\

Before going to the Loop Vertex Expansion, let us introduce a few notations. First of all, in the following, we will denote by $\mathcal T_n$ be the set of all labelled trees with $n$ vertices. To a tree $T\in \mathcal T_n$, we associate the symmetric $n \times n$ matrix $W^{T}(u)$ with diagonal
entries equal to $1$ and off diagonal ones $W_{ij}^{F}(u)=W_{ji}^{F}(u)=w_{ij}^{F}(u)$
as given by eq.~\eqref{eq:BKAR}. Then, the LVE is written in terms of
$$C^n_k = \big\{ (i_1,...,i_k)\in
\{1,...,n\}^k  \mid \text{for all } a,b \in \{1,...,k\},a \neq b \Rightarrow
i_a \neq i_b \big\} \; .$$ 
The notation comes from the fact that $C_k^n$ is the
configuration space of $k$ particles on the discrete $n$-point space. With this notation at hand, the LVE allows us to express to cumulants as a sum over \textit{ciliated
  trees}, for which we adopt the following convention:
  \begin{equation}
    \text{couples $(T,\mathfrak{c})$
made of a tree $T$ (with $n$ vertices)
and \emph{cilia} $\mathfrak{c}\in C^n_k$ are denoted by $T_{\mathfrak{c}}$ }.\tag{$\star$}  \label{T_c} 
\end{equation}

For $i\in\{1,...,n\}$, we also denote $d_i(T_{\mathfrak{c}})=c(i)+d_i(T)$ with $c(i)=\textbf{1}_{i\in \mathfrak{c}}$ the coordination (or degree) of the vertex $i$ in $T_{\mathfrak{c}}$ including the cilia. The Loop Vertex expansion of the cumulants is given by the following proposition:
\begin{proposition}[Loop Vertex expansion of the cumulants] For all $1\leq k\leq n$ and $\psi \in (\theta- \pi/ 2,\theta+\pi /2)$, the cumulants are given by the series
\begin{multline*}
\mathfrak{K}_{\psi}^{2k}(g,\epsilon) = 2^{k-1} \sum_{n \geq k}
\frac{1}{n!} {\left(\frac{-\Pi g}{2}\right)}^{n-1}\\\times
\sum_{T_{\mathfrak{c}}\in \mathcal{T}_n\times C^n_k}
\int du_{T} \int d\mu_{\epsilon
  e^{-\imath \psi} W^{T}(u)}(\sigma) \bigg\{ \prod_{i=1}^n 
(d_i(T_{\mathfrak{c}}) -1)! \,
\liftR^{d_i(T_{\mathfrak{c}})}(\sigma^{(i)} e^{\imath \frac \psi2}, g) \bigg\} \;,
\end{multline*}
for all $(g,\epsilon)\in\Sigma\times\rmH $ such that it converges (in particular for $(g,\epsilon)\in\rmI(\R_+^*)\times\R_+^* $).
\end{proposition}
\begin{proof}
To perform the LVE, we first expand the partition function with sources as expressed in
eq.~\eqref{eq-Ztiltcontour}. From now on, we fix
$\psi \in (\theta- \pi/ 2,\theta+\pi /2)$ and for $(g,\epsilon)\in\mathfrak{B} $ we start from:
\begin{align}\nonumber
     Z_{\psi}(g,\epsilon ;J) &= \int d\mu_{\epsilon e^{-\imath
         \psi}}(\sigma) \; e^{\frac{1}{2\epsilon}\ln \liftR(\sigma e^{\imath\frac
         \psi2},  g) + \frac{1}{2\epsilon}\liftR(\sigma e^{\imath\frac
         \psi2},  g)\,\norm{J}^2}\,.
\end{align}
We expand the exponential inside the integral and, using Fubbini's theorem, we exchange the sum and the
integral to obtain:
\begin{equation}\label{partfunc}
    Z_{\psi}(g,\epsilon ;J) = \sum_{n=0}^\infty \frac 1{(2\epsilon)^nn!} \int
    d\mu_{\epsilon e^{-\imath \psi }}(\sigma)
    \big[\ln \liftR(\sigma e^{\imath\frac
         \psi2},  g) +
     \liftR(\sigma e^{\imath\frac
         \psi2},  g)\norm{J}^2\big]^n \;.
\end{equation}
 The use of Fubbini's theorem is justified by the following lemma:
\begin{lemma}\label{thm-ZZtilde} Let $a\in (0,1/2]$, $(g,\epsilon)\in\Sigma\times \rmH$ and $\psi\in (\theta-\pi/2,\theta+\pi/2)$. Then, if $(g,\epsilon)\in\mathfrak A^+_\psi$ (see eq.~\ref{Apsiplus}), there exist two non negative reals $C_1, C_2$ independent of $\modulus{g}$ and $\modulus{\epsilon}$ such that for $n$ large enough:
\begin{equation*}
   A_n:= \frac{1}{(2\modulus{\epsilon})^nn!}\bigg| \int d\mu_{\epsilon e^{-\imath \psi}}(\sigma)
    \big[{\ln\liftR(\sigma e^{\imath\frac
         \psi2},  g) +\liftR(\sigma e^{\imath\frac
         \psi2},  g)\norm{J}^2}\big]^n \bigg|\leq \frac{C_1^n}{\modulus{\epsilon}^n n!}+\frac{C_2^n \modulus{g}^{an}}{{(\modulus{\epsilon}^n n!)}^{(1-a)}}\,.
\end{equation*}
In particular, at fixed $\psi\in (\theta-\pi/2,\theta+\pi/2)$, the sum of the $A_n$'s has infinite radius of convergence in both $\modulus{g}$ and $1/\modulus{\epsilon}$ for all $(g,\epsilon)\in\mathfrak A_\psi^+$.
\end{lemma}
\begin{proof}It is convenient to perform the change of variable $\sigma \rightarrow \sigma \sqrt{\modulus{g}}$ so that $A_n$ rewrites:
  \begin{equation*}
\frac{1}{(2\modulus{\epsilon})^n n!}\bigg| \int d\mu_{\modulus{g}\epsilon e^{-\imath \psi}}(\sigma)[\ln \liftR(\sigma e^{\imath \frac\psi2},\frac{g}{\modulus{g}} )+\liftR(\sigma e^{\imath \frac\psi2},\frac{g}{\modulus{g}} )
\norm{J}^2]^n \bigg|.
\end{equation*}
 Then, thanks to the bound in eq.~\eqref{resbound}, if $\modulus{\psi+\varphi}<\pi$,
$\modulus{\liftR(\sigma e^{\imath \frac\psi2},\frac{g}{\modulus{g}} )}
\cdot \norm{J}^2 \leq
\frac{\norm{J}^2}{\cos{(\frac{\psi+\varphi}{2})}}$. Furthermore:
\begin{equation*}
\big|
  {  \ln{ \liftR(\sigma e^{\imath \frac\psi2},\frac{g}{\modulus{g}} )} } \big|\leq
  \big[ \big(-\frac{1}{2}\ln\lvert
    \liftR^{-2}(\sigma e^{\imath \frac\psi2},\frac{g}{\modulus{g}} ) \rvert\big)^2+\pi^2 \big]^{1/2}.
  \end{equation*}
Then, using $\lvert \liftR^{-2}(\sigma e^{\imath \frac\psi2},\frac{g}{\modulus{g}} ) \rvert = 1 +2\sin{\frac{\psi+\varphi}{2}}\modulus{\sigma}+\sigma^2\leq(1+\modulus{\sigma})^2$, we get $\ln{\lvert \liftR^{-2}(\sigma e^{\imath \frac\psi2},\frac{g}{\modulus{g}} ) \rvert}\leq \ln{(1+\modulus{\sigma})^2}\leq 2(1+\modulus{\sigma})^{2a}$ for $0<a<1/2$ implying $ |
  {  \ln{\liftR(\sigma e^{\imath \frac\psi2},\frac{g}{\modulus{g}} )} }|\leq  \sqrt{(1+\modulus{\sigma})^{4a}+\pi^2 }  \le (1+\modulus{\sigma})^{2a}+\pi $ as $\pi>1$ so that:
\begin{align*}
A_n &\leq\frac{1}{(2\modulus{\epsilon})^n n!\sqrt{\cos(\psi-\theta)}}\int d \mu_{\frac{\modulus{g\epsilon}}{\cos(\psi-\theta)}}(\sigma) \big[(1+\modulus{\sigma})^{2a}+\pi +\norm{J}^2\cos^{-1}{(\frac{\psi+\varphi}{2})}\big]^n\\
&\leq \frac{1}{(2\modulus{\epsilon})^n n!\sqrt{\cos(\psi-\theta)}}\int d \mu_{\frac{\modulus{g\epsilon}}{\cos(\psi-\theta)}}(\sigma) \sum_{k=0}^n\binom{n}{k}(1+\modulus{\sigma})^{2ak} \big[\pi +\norm{J}^2\cos^{-1}{(\frac{\psi+\varphi}{2})}\big]^{n-k}\\
&\leq \frac{\big[\pi +\norm{J}^2\cos^{-1}{(\frac{\psi+\varphi}{2})}\big]^{n}}{\modulus{\epsilon}^n n!\sqrt{\cos(\psi-\theta)}}\int d \mu_{\frac{\modulus{g\epsilon}}{\cos(\psi-\theta)}}(\sigma) (1+\modulus{\sigma})^{2an}\,,
\end{align*}
where we used the fact that $\pi +\norm{J}^2\cos^{-1}{(\frac{\psi+\varphi}{2})}$ and $1+\modulus{\sigma}$ are greater than one. At this stage, rewriting:
\begin{align*}
\int d \mu_{\frac{\modulus{g\epsilon}}{\cos(\psi-\theta)}}(\sigma) (1+\modulus{\sigma})^{2an} &= \int_{\modulus{\sigma}<1} d \mu_{\frac{\modulus{g\epsilon}}{\cos(\psi-\theta)}}(\sigma) (1+\modulus{\sigma})^{2an} +\int_{\modulus{\sigma}>1} d \mu_{\frac{\modulus{g\epsilon}}{\cos(\psi-\theta)}}(\sigma) (1+\modulus{\sigma})^{2an} \\
&\leq \int_{\modulus{\sigma}<1} d \mu_{\frac{\modulus{g\epsilon}}{\cos(\psi-\theta)}}(\sigma) 2^{2an} +\int_{\modulus{\sigma}>1} d \mu_{\frac{\modulus{g\epsilon}}{\cos(\psi-\theta)}}(\sigma) (2\modulus{\sigma})^{2an} \\
&\leq 2^{2an} \bigg(\int_{\R} d \mu_{\frac{\modulus{g\epsilon}}{\cos(\psi-\theta)}}(\sigma) +\int_{\R} d \mu_{\frac{\modulus{g\epsilon}}{\cos(\psi-\theta)}}(\sigma) \modulus{\sigma}^{2an}\bigg)\\
&\leq 4^{an}\bigg(1+\frac{1}{2\sqrt{\pi}}\big[\frac{\cos{(\psi-\theta)}}{2\modulus{g\epsilon}}\big]^{1/2}\int_{\R_+}e^{-\cos{(\psi-\theta)}\frac{t}{2\modulus{g\epsilon}}}t^{an-1/2}dt\bigg)\\
&\leq 4^{an}\bigg(1+\frac{1}{2\sqrt{\pi}}\big[\frac{2\modulus{g\epsilon}}{\cos{(\psi-\theta)}}\big]^{an}\Gamma\big(an+\frac{1}{2}\big)\bigg)\,,
\end{align*}
and using the asymptotic of the gamma function we get $A_n \leq \frac{C_1^n}{\modulus{\epsilon}^n n!}+\frac{C_2^n \modulus{g}^{an}}{{(\modulus{\epsilon}^n n!)}^{(1-a)}}$.\end{proof}

We now use the copies trick as stated in Lemma~\ref{replica} (see  Appendix~\ref{replicatr} for the proof). With our integral notations,
Lemma~\ref{replica} rewrites:
\begin{lemma}[The copies trick]
Let $n$ be a positive integer, $z\in\C$ with $\Re z>0$ and 
 $F \in L^n(\R,\mu_{\textfrac{\lvert z\rvert^2}{\Re z} 
})$ a $\C$-valued function. Then $F^{\otimes n }:\R^n\to\C$, $X=(X_i)_{1\leq i\leq n} \mapsto
\prod_{1=n}^n F(X_i)$ is in $L^1(\R^n,\mu_{{\textfrac{\lvert
      z\rvert^2 \mathbbm{1}_n}{\Re z}} })$ and furthermore we
have:
\begin{align}\label{eq:replica}
 \int d\mu_{z}(x) F^n(x) = \int d\mu_{z  \mathbbm{1}_n}(X)  F^{\otimes n}(X)=\int d\mu_{z  \mathbbm{1}_n}(X) \prod_{i=1}^n F(X_{i}) \;.
\end{align}
\end{lemma}
This lemma produces for an integration variable $\sigma$, $n$ 
variables of integration that we call the copies of $\sigma$ and that
we denote by $(\sigma^{(i)})_{1\leq i \leq n}$, where we use parenthesis
to avoid the confusion with the O($N$) indices. With this notation, using eq.~\eqref{eq:replica} in 
eq.~\eqref{partfunc} we obtain that:
\begin{align} 
&Z_{\psi}(g,\epsilon ; J) =\sum_{n \geq 0} \frac{1}{(2\epsilon)^nn!} \int
d\mu_{\epsilon e^{-\imath \psi} \mathbbm{1}_{n}}(\sigma) \prod_{i=1}^n
\bigg\{ \ln \liftR(\sigma^{(i)}e^{\imath \frac\psi2}, g) +
\liftR(\sigma^{(i)}e^{\imath \frac\psi2}, g)\norm{J}^2 \bigg\}\label{square_brackets}
\\ 
={}& \sum_{n \geq 0} \frac{1}{(2\epsilon)^nn!} \left[ \exp\Big(\frac{\epsilon
    e^{-\imath \psi}}{ 2} \langle \partial, \partial \rangle_{X(x)}
  \Big) \prod_{i=1}^n \bigg\{
  \ln\liftR(\sigma^{(i)} e^{\imath \frac \psi2}, g) +
  \liftR(\sigma^{(i)} e^{\imath \frac \psi2}, g)\norm{J}^2 \bigg\}
  \right]_{\sigma^{(i)}=0,x_{ij}=1}\,,\nonumber
\end{align}
where $[X(x)]_{ii}=1$ and $[X(x)]_{ij}=x_{ij}$ for $i\neq j$. 

We now rewrite the above expansion as a sum indexed by forests over labelled vertices of analytic functions and consequently its logarithm as a sum indexed by trees over labelled vertices of analytic functions. The crucial point is that at order $n$ the number of 
such trees is of order $O(1)^n n!$, much less than the number of Feynman diagrams which is of order
$O(1)^n(2n)!$, and the expansion is convergent.

This rewriting is obtained thanks to the BKAR formula \cite{Brydges1987aa,Abdesselam1995aa}, which will be applied to the $x=(x_{ij})$ parameters.
\begin{lemma}[BKAR formula] \label{thm:BKAR}
Let $f: \mathbb{R}^{\frac{n(n-1)}{2}}\rightarrow {\mathbb C}$ be a
smooth function. If $x\in\R^{\frac{n(n-1)}{2}}$, we denote its
components by $x_{ij}$ for $i,j=1,2,\dots,n$, $i<j$. Let $F$ be a
forest with vertex set $V(F)=\{1,2,\dotsc,n\}$. If there is an edge between vertices $i$ and $j$ 
($i<j$), we write $(i,j)\in E(F)$. If both $i$ and $j$ belong to
the same connected component of $F$, we let $P_{i\leftrightarrow
  j}^{F}$ stand for the unique path in $F$ between $i$ and
$j$. Then we have:
\begin{equation}
f(x) \big|_{x=\mathbf{1}}=\sum_{F\in \mathcal F_n}\Big(\prod_{(i,j)\in E(F)} \int_{0}^{1} du_{ij}\Big) \,
 \frac{\partial^{|E(F)|} f(x) }{\prod_{(l,m)\in E(F)} \partial x_{lm}}  \bigg| _{x=w^{F}} \,,
\label{BKARformula}
\end{equation}
where $\mathbf{1}$ is the $\frac{n(n-1)}2$-vector with  all the components equal to $1$,  $\mathcal F_n$ is the set of all forests with $n$ labelled vertices, 
$|E(F)|$ is the number of edges of $F$, and
$w^{F}\in\R^{\frac{n(n-1)}{2}}$ is given by:
\begin{equation} \label{eq:BKAR}
    w^{F}_{ij}\defi\begin{cases} \inf_{(k,l)\in P_{i\leftrightarrow
        j}^{F}} u_{kl}&\text{if $i$ and $j$ belong to the same
      component of $F$,}\\ 0 &\text{otherwise.}
\end{cases} 
\end{equation} 
\end{lemma}

Notice that $\prod_{(a,b) \in E(F)}\partial_{x_{ab}} = {\epsilon}^{\lvert E(F) \rvert}e^{-\imath \psi \lvert E(F) \rvert} \prod_{(ab) \in E(F)} \partial_{\sigma^{(a)}}\partial_{\sigma^{(b)}}$ when acting on the quantity in square brackets in \eqref{square_brackets}.
Thus, by  the BKAR formula~\eqref{BKARformula}, the partition function with sources rewrites as
\begin{multline*}
Z_{\psi}(g,\epsilon ; J) = \sum_{n \geq 0} \frac{1}{(2\epsilon)^nn!} \sum_{F\in \mathcal F_n}
\int du_{F} \int d\mu_{\epsilon e^{-\imath \psi} W^{F}(u)}(\sigma){\epsilon}^{\lvert E(F) \rvert} e^{-\imath \psi \lvert E(F) \rvert}\\
\times\prod_{(a,b) \in E(F)}
     \partial_{\sigma^{(a)}}\partial_{\sigma^{(b)}}
     \prod_{i=1}^n
     \bigg\{ \ln\liftR(\sigma^{(i)} e^{\imath \frac \psi2}, g) +
    \liftR(\sigma^{(i)} e^{\imath \frac \psi2}, g)\norm{J}^2 \bigg\} \;,
\end{multline*}
where we recall that for $F\in \mathcal F_n$, $W^{F}(u)$ is the symmetric $n \times n$ matrix with diagonal
entries equal to $1$ and off diagonal ones $W_{ij}^{F}(u)=W_{ji}^{F}(u)=w_{ij}^{F}(u)$
as given by eq.~\eqref{eq:BKAR}, $d u_{F}=\prod_{(ij)\in{F}}
\int_{0}^{1} d u_{ij}$. The matrix $W^{F}(u)$ is positive.

The logarithm of the partition function with sources is the generating
function of the cumulants. Now that the partition function with sources
is expressed as a sum over forests with the amplitude of a forest factorizing over the trees of this forest, its logarithm writes as a sum over trees. As $\lvert E(T) \rvert = |V(T)|-1$, we get:
\begin{multline*}
\epsilon  \ln{Z}_{\psi}(g,\epsilon ; J)= \sum_{n \geq 1} \frac{e^{-\imath (n-1) \psi}}{2^n n!} \sum_{T \in \mathcal T_n}
  \int du_{T} \int d\mu_{\epsilon e^{-\imath \psi}   W^{T}(u)}(\sigma)  \prod_{(a,b) \in E(T)} \partial_{\sigma^{(a)}}\partial_{\sigma^{(b)}} \sum_{k=0}^n \frac{{\norm{J}^{2k}}}{k!} 
  \\ \times \sum_{\substack{1 \leq i_1,...,i_k \leq n
      \\ a \neq b
    \Rightarrow i_a \neq i_b}} \prod_{c=1}^k
\liftR(\sigma^{(i_c)}e^{\imath \frac \psi2}, g) \prod_{j \neq i_1,...,i_k}
\ln{\liftR(\sigma^{(j)}e^{\imath \frac \psi2}, g})\,,
\end{multline*}
were we recall that $\mathcal{T}_n$ stands for the set of all trees with $n$ labelled vertices. At this stage, we can rewrite the sum above in terms of $C^n_k$ and of \textit{ciliated trees} (recall the convention \eqref{T_c}). Indeed, repeatedly using the fact
that $\partial_{\sigma}^k \ln \liftR (\sigma e^{\imath \frac \psi2}, g) =
{(\imath e^{\imath \frac \psi2}\liftsqrt{ g})}^k (k-1)! \, \liftR^k(\sigma e^{\imath \frac \psi2}, g)$, $ \partial_{\sigma}^k \liftR(\sigma e^{\imath \frac \psi2}, g) =
{(\imath e^{\imath \frac \psi2}\liftsqrt{ g})}^k k! \, \liftR^{k+1}(\sigma e^{\imath \frac \psi2}, g)$,  $(\liftsqrt{g})^2=\Pi g$ and the combinatorial identity $ \sum_i d_i = 2(n-1) $, the
logarithm of the partition function with sources becomes:
\begin{multline*}
\epsilon \ln{Z}_{\psi}(g,\epsilon ; J) = \frac{1}{2} \sum_{n \geq 1}
\frac{1}{n!} {\left(\frac{- \Pi g}{2} \right)}^{n-1} \sum_{k=0}^n \frac{
  \norm{J}^{2k} }{k!}  \sum_{T_{\mathfrak{c}}\in \mathcal{T}_n\times C^n_k} \int du_{T} \int
d\mu_{\epsilon e^{-\imath \psi} W^{T}(u)}(\sigma)\\
\times
\prod_{i=1}^n \big\{ (d_i(T_{\mathfrak{c}}) -1)! \,
\liftR^{d_i(T_{\mathfrak{c}})}(\sigma^{(i)} e^{\imath \frac \psi2}, g) \big\}.
\end{multline*}

Using 
$\frac{\partial^{2k}}{\partial_{J_{a_1}}...\partial_{J_{a_{2k}}}}
\norm{J}^{2k} = 2^k k!\sum_{\pi\in P_2(2k)} \prod_{(a_i,a_j) \in \pi}
\delta_{a_i,a_j}$, we obtain the following expression, that holds true as long as both the individual integrals converge and the overall series is convergent (recall that with the convention \eqref{T_c}, $T$ is
the tree $T_{\mathfrak c}$ without cilia):
\begin{multline}\label{cumul}
\mathfrak{K}_{\psi}^{2k}(g,\epsilon) = 2^{k-1} \sum_{n \geq k}
\frac{1}{n!} {\left(\frac{- \Pi g}{2}\right)}^{n-1}\\\times
\sum_{T_{\mathfrak{c}}\in \mathcal{T}_n\times C^n_k}
\int du_{T} \int d\mu_{\epsilon
  e^{-\imath \psi} W^{T}(u)}(\sigma) \bigg\{ \prod_{i=1}^n 
(d_i(T_{\mathfrak{c}}) -1)! \,
\liftR^{d_i(T_{\mathfrak{c}})}(\sigma^{(i)} e^{\imath \frac \psi2}, g) \bigg\} \;,
\end{multline}
which concludes the proof.
\end{proof}

Our first main theorem concerns the domain in $(g,\epsilon)$ in which the rescaled cumulants are
analytic in both variables. 
\begin{theorem}[Main Theorem 1: Analyticity]\label{THM1}
For all $k \geq 1$, the cumulant of order $2k$ of
the quartic $\grpO(N)$-vector model, $\mathfrak{K}^{2k}(g,\epsilon)$, as expressed by the series~\eqref{cumul}, is analytic in $g$ and $\epsilon$ on the domain $\mathfrak{C}$ consisting in all the couples $(g,\epsilon)
\in \Sigma\times\rmH$ 
such that there exists $\psi\in(-\pi,\pi)$ for which the following inequalities hold:
\begin{subequations}%
\label{conds}
 \begin{align}
\lvert g \rvert & <\frac{1}{4} (1+\cos{(\liftarg g+\psi)})
\sqrt{\cos{(\psi-\arg \epsilon)}} \label{conds1} \; ,
\\ \lvert \liftarg g+\psi \rvert
& < \pi \label{conds2} \; ,
\\ \modulus{\psi-\arg \epsilon} & <\frac{\pi}{2} \label{conds3}   \; .
\end{align}
\end{subequations}%
\end{theorem}
The proof of this theorem is given is Section~\ref{sec4}.

\begin{corollary}[Domain as a Riemann sheet]\label{corol1}
At fixed $g\in\Sigma$, the domain of analyticity in $\epsilon$ is independent of its modulus, and contains all $\epsilon\in\rmH$ such that $ - 3\pi/2 - \liftarg g < \arg \epsilon < 3\pi/2 - \liftarg g $. In particular, for all
$\lvert \epsilon \rvert \geq 0$, $\theta \in
(-\pi/2,\pi/2)$, and $\varphi \in (-\pi,\pi)$, $((\modulus{g},\varphi),\lvert \epsilon \rvert e^{\imath
  \theta})$ belongs to $\mathfrak{C}$ if $\lvert g\rvert$ is small
enough (see discussion in Remark~\ref{rk:domain} for how small ``small enough'' is) and for such $g$, $\mathfrak{C}$ includes a Sokal disk
in the $\epsilon$-plane (see Remark~\ref{remarksokaldisk}) of an
arbitrary, positive radius.
\end{corollary}
\begin{proof}
The two conditions on $\varphi$ and $\theta$ read:
\begin{equation*}
    \begin{cases}
\lvert \varphi+\psi \rvert < \pi \\ \lvert \psi -\theta \rvert < \frac{\pi}{2}
    \end{cases} \Leftrightarrow
     \begin{cases}
-\pi< \varphi+\psi < \pi \\ -\frac{\pi}{2}< \theta-\psi  < \frac{\pi}{2}
    \end{cases} \Rightarrow -\frac{3\pi}{2}< \varphi+\theta  <\frac{3\pi}{2}
     \Leftrightarrow  - \frac{3\pi}{2} -  \varphi < \theta < \frac{3\pi}{2} -  \varphi \; ,
\end{equation*}
and observing that $\bigcap_{\varphi\in(-\pi,\pi)} (- 3\pi/2 -  \varphi ,
3\pi/2 -  \varphi)=(-\pi/2,\pi/2)$ we conclude.
\end{proof}
\begin{theorem}[Main Theorem 2: Borel summability]\label{THM2}
For small $\alpha > 0$ we define the subdomain
$\mathfrak{C}_{\alpha}$ of $\mathfrak{C}$ made of all couples $(g,\epsilon)
\in \Sigma\times\rmH$
such that there exists $\psi \in (-\pi,\pi)$ for which the following inequalities hold:
\begin{subequations}%
\label{condsa}
\begin{align}
  \lvert  g \rvert & <\frac{1}{4} (1+\cos{(\liftarg g+\psi)}) \sqrt{\cos{(\psi-\arg \epsilon)}} \,(1-\alpha), \label{condsa1}
  \\ \lvert \liftarg g+\psi \rvert & < \pi\, (1-\alpha), \label{condsa2}\\  \lvert \psi-\arg \epsilon \rvert &< \frac{\pi}{2}\, (1-\alpha).\label{condsa3}
\end{align}
\end{subequations}%
For all $k \geq 1$, the rescaled cumulant of order $2k$ of the $\grpO(N)$-vector model
$\mathfrak{K}^{2k}(g,\epsilon)$ is Borel summable in $\epsilon$ along the positive real axis for $g$ inside a non trivial domain (see Remark~\ref{rk:domain}) and for $g$ inside this domain they can be computed as the Borel sum of their large $N$ expansion.
\begin{proof}
The proof of this theorem follows from Corollary~\ref{corol1}
and from the following lemma, proven in Section~\ref{sec5}:

\begin{lemma}\label{lemmaboundrest}
For small $\alpha>0$, for all $k\geq1$ and for all $(g,\epsilon)\in
\mathfrak{C}_\alpha$, there exists two constants $C_\alpha >0$ and
$K_\alpha >0$ independent of $g$ and $\epsilon$ but depending on $\alpha$ such that the Taylor rest term of order $q$ in the $\epsilon$ expansion of the
cumulant, denoted by $R^{2k}_q(g,\epsilon)$ (see eq.~\eqref{restterm} for a 
closed expression of this rest term), obeys the following bound for $q$ large enough: 
\begin{equation*}
\lvert R^{2k}_q(g,\epsilon) \rvert \lesssim_k C_\alpha K_\alpha^q
       {\lvert \epsilon \rvert}^{q} q!\;.
\end{equation*}
\end{lemma} 
This bound together with Corollary~\ref{corol1} prove that the cumulants
verify the hypotheses of the Nevanlinna Sokal theorem \cite{Sokal1980aa} uniformly in $g$ 
(for completeness we recall the relevant version of this theorem in Thm.~\ref{thm:Sokal}). 
\end{proof}

\end{theorem}

\begin{remark}\label{rk:domain}
We now wish to visualize the domain $\mathfrak{C} \subset \Sigma\times
\rmH $ (or $\mathfrak{C}_{\alpha} $ for $\alpha\to 0$). Let us go to the $\C$-plane of $\Pi g$ 
and look for the curve $\rho(\varphi)$ defined by:
$$ \varphi \mapsto \rho(\varphi):= \sup \{ \, |g| : \text{there is a
  $\psi=\psi(\theta)$ such that $|g|e^{\imath\varphi}\in \mathrm{pr}_1
  \mathfrak{C}^{\psi(\theta)}$ for all $ \theta \in (-\pi/2,\pi/2 )$} \, \} .$$ 
    Here, $\mathfrak C^\psi$ consists of
    the points $(g,\epsilon)$ verifying eqs. \eqref{conds}
    for a given $\psi$, and pr$_1$  is the projection
to the first $\C$-factor (or the $g$-plane), so that, in particular, the conditions:
$$ 
    \lvert \varphi+\psi(\theta) \rvert < \pi,\quad \lvert
    \psi(\theta)-\theta \rvert < \frac{\pi}{2},\quad \text{ and }\quad\lvert g
    \rvert <\frac{1}{4} \big(1+\cos{[\varphi+\psi(\theta) ]}\big)
    \sqrt{\cos{(\theta-\psi(\theta) )}} \;,
$$ 
must hold. The visualization of this curve is easier for a linear choice
    $\psi_\xi(\theta)=  \xi \cdot \theta$, where $0<\xi < 1$ is a new
    parameter. Denoting by $\rho_\xi(\varphi)$ the curve for this particular
    choice of $\psi=\psi_\xi$, namely:
\begin{align}
     \rho_\xi(\varphi) :=
\sup \{\, |g| :  \,\,  |g|e^{\imath\varphi}\in \mathrm{pr}_1 \mathfrak{C}^{\psi_\xi(\theta)}
\text{ for all }  \theta \in (- \pi/2,\pi/2 )  \} \;, \label{linearpsi}
\end{align}
this curve can be visualized (see Fig. \ref{fig:Cpsinus}).
\begin{figure}[ht]
\centering
 \includegraphics[width=.35\linewidth]{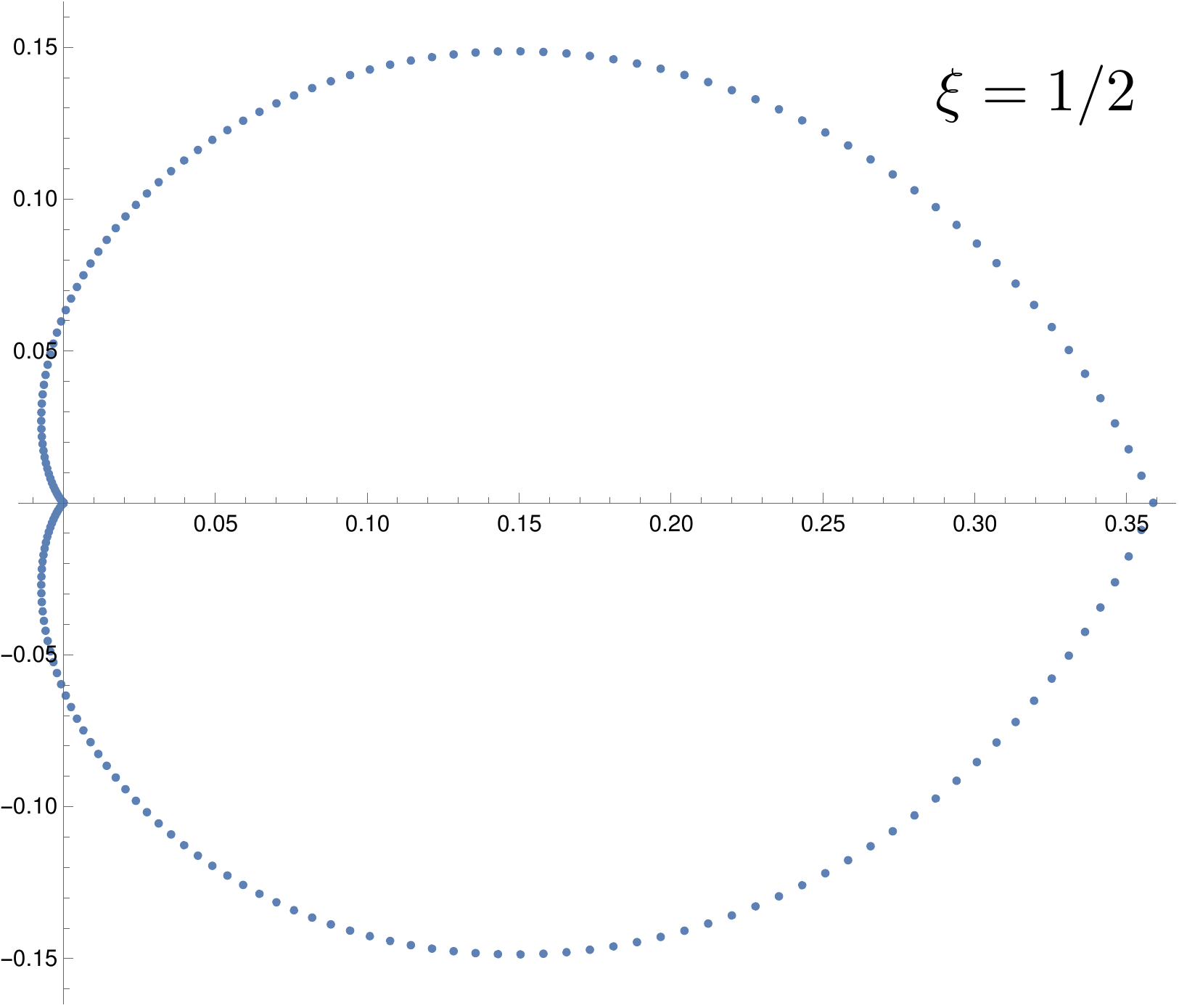}
 \raisebox{-2.99ex}{\includegraphics[width=.35\linewidth]{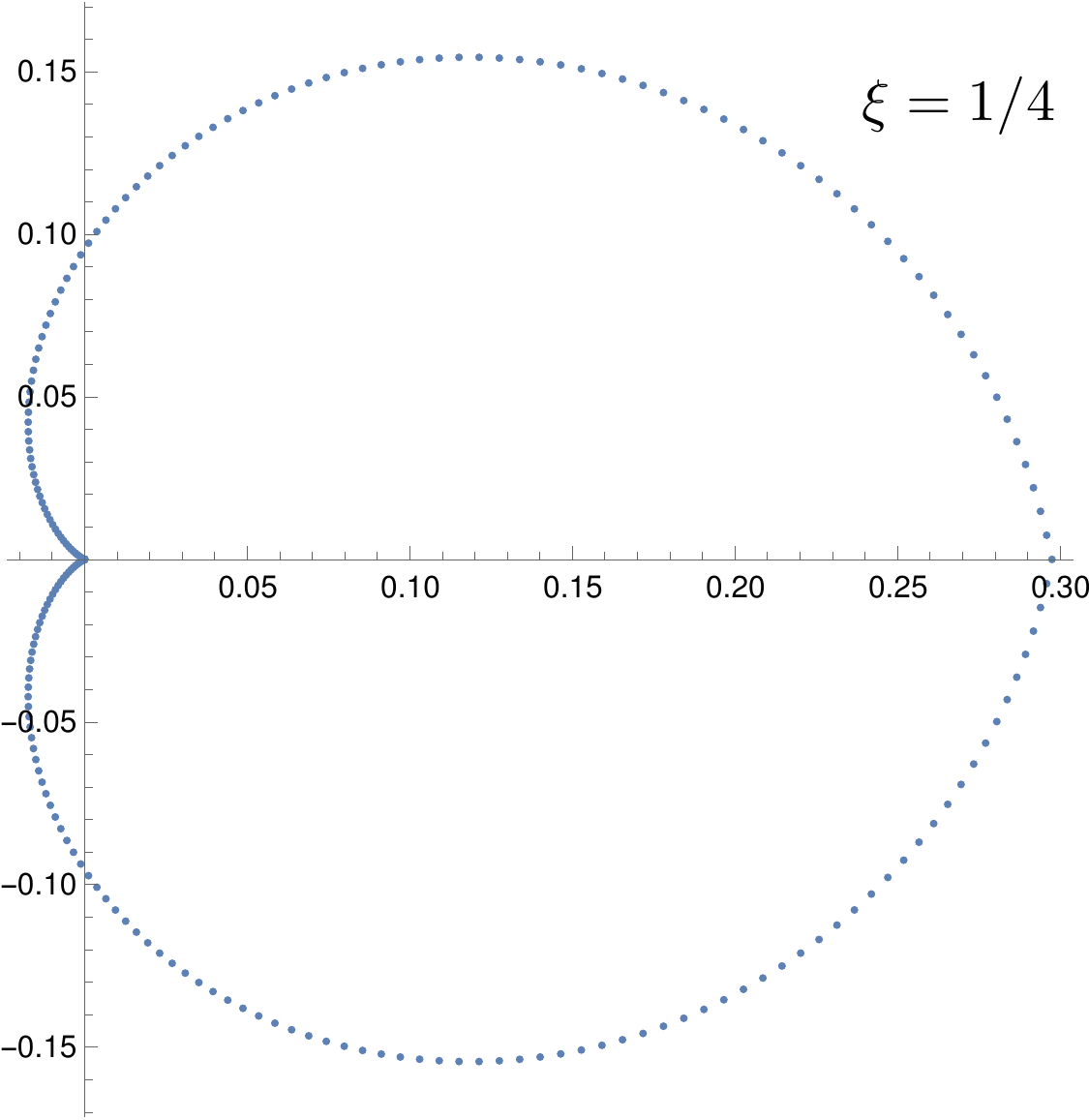}}\\[2ex]
       \includegraphics[width=.35\linewidth]{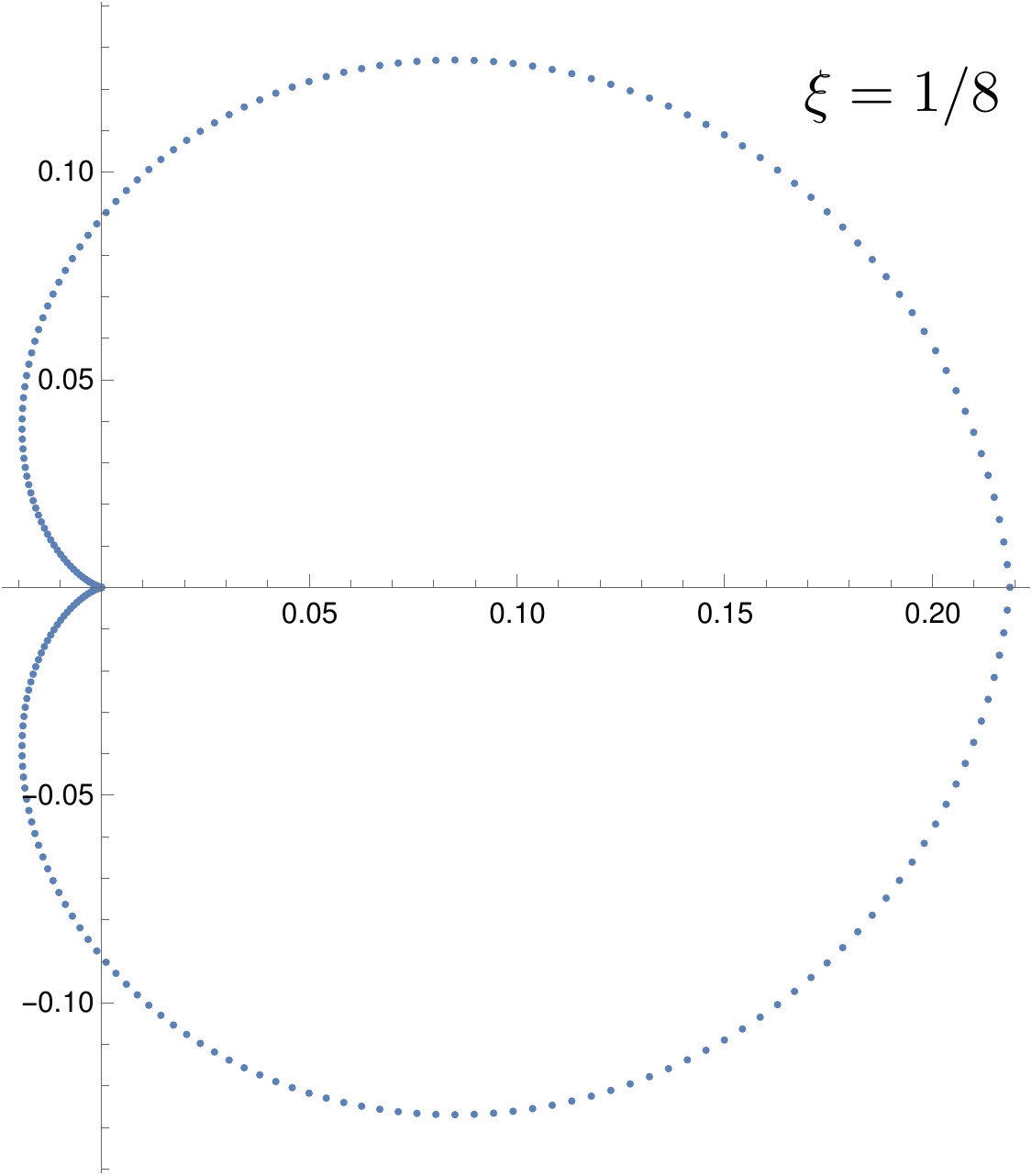}  \quad
\raisebox{4.9ex}{\includegraphics[width=.35\linewidth]{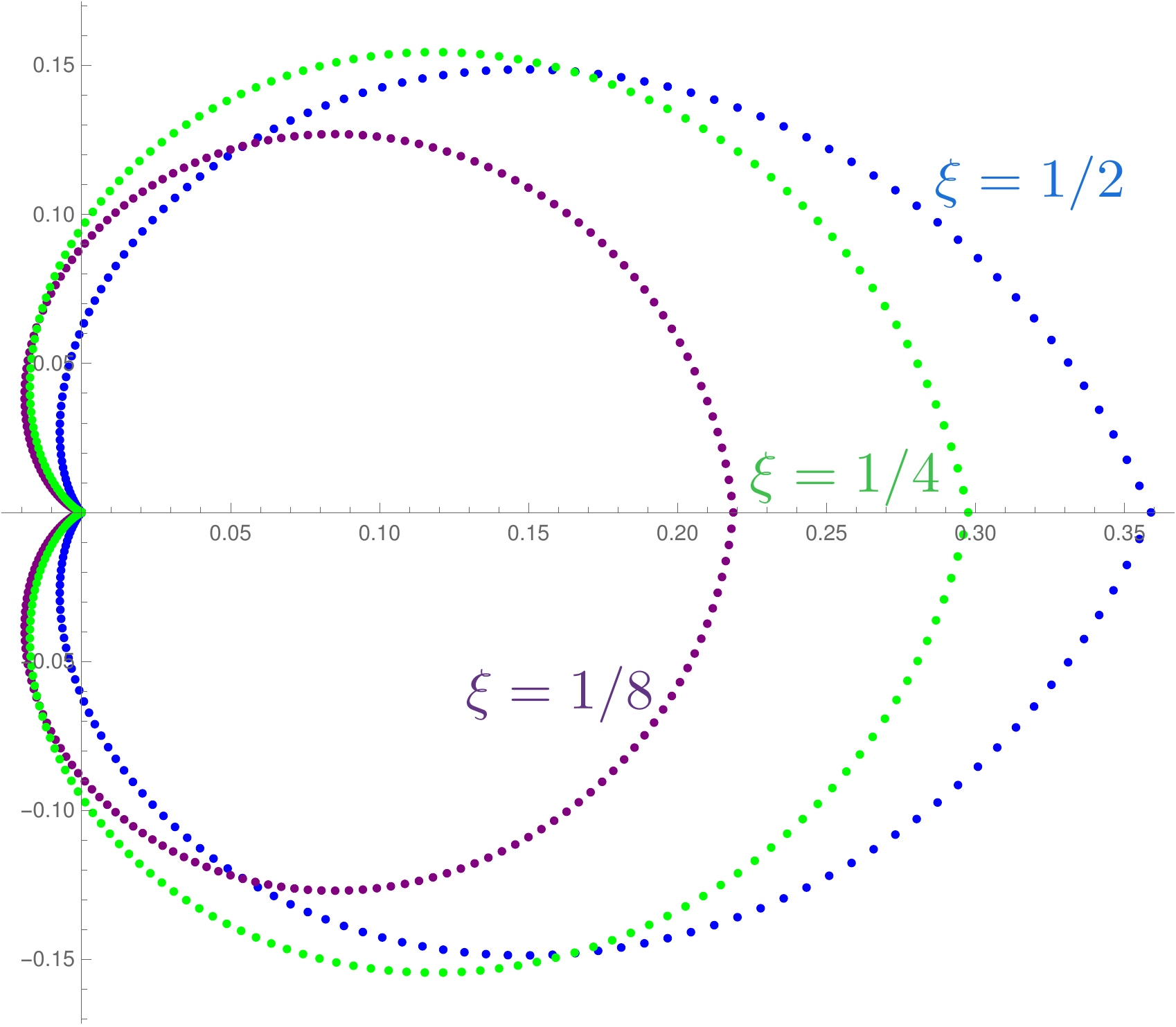}}
\vspace{.3cm}
\caption{In the first three panels, we show the (discretized) curves
  $\rho_\xi(\varphi)$ given by eq. \eqref{linearpsi} for the values of
  $\xi=1/2,1/4,1/8$. The last panel shows the superposed domains.\label{fig:Cpsinus}}
\end{figure}
\end{remark}

\section{The bounds and the domain of analyticity}\label{sec4}

This section is dedicated to:
\begin{proof}[Proof of Theorem~\ref{THM1}]
From now on we fix $\psi \in (\theta- \pi/ 2,\theta + \pi/
2)$ and bound $\mathfrak{K}_{\psi}^{2k}(g,\epsilon)$ (as
expressed in eq.~\eqref{cumul}) thanks to eq.~\eqref{resbound} and the following lemma, 
proven in Appendix~\ref{app:compbound}:
\begin{lemma}\label{complexgaussianboundmain}
Let $n$ be a positive integer, $z\in\C$ with $\Re z>0$, $C \in
M_n(\R)$ symmetric positive matrix and $F \in
L^1(\R^n,\mu_{{\textfrac{\lvert z\rvert^2 C}{\Re z}}})$ a $\C$-valued function. Then:
\begin{equation}\label{GaussianComplex}
 \left| \int d\mu_{zC}(X)F(X) \right| = \left| \left[ e^{ \frac{z}{2} \langle \partial,\partial \rangle_{C}}
    F(X) \right]_{X=0}
 \right| \le \frac{1}{ \cos^{n/2}(\arg z)  }  \sup_{X \in \mathbb{R}^n} | F(X) |\,.
\end{equation}
\end{lemma}
We apply this lemma with $z=\epsilon e^{-\imath \psi}$ and
$C=W^{T}(u)$. Then, we bound the integration over the
$u_{T}$ parameters by one. Finally, we also have to
notice that $\frac{1}{\cos{(\frac{\varphi+\psi}{2})}}$ appears at the
power $\sum_{i=1}^n d_{i}(T_{\mathfrak{c}}) = k + 2(n-1)$. Thus,
\begin{multline*}
\lvert \mathfrak{K}_{\psi}^{2k}(g,\epsilon) \rvert \lesssim_k 
\frac{1}{\sqrt{\cos{(\psi-\theta)}} }  \frac{1}{\cos^k{(\frac{\varphi
      +\psi}{2})}} \\
      \times \sum_{n \geq k} \frac{1}{n!}{\left(\frac{\lvert g
    \rvert}{2 \cos^2{(\frac{\varphi+\psi}{2})}
    \sqrt{\cos{(\psi-\theta)}} }\right)}^{n-1}
\sum_{T_{\mathfrak{c}}\in \mathcal{T}_n\times C^n_k} \prod_{i=1}^n (d_{i}(T_{\mathfrak{c}})-1)! \; ,
\end{multline*}
We conclude using combinatorial arguments, that are gathered
in the following lemma:
\begin{lemma}\label{lemmacombina1}
For all $n \geq k$, the sum over $k$-ciliated (i.e. $\mathfrak c \in C_k^n$) trees with $n$-vertices verifies
\begin{equation*}
    \frac{1}{n!}
\sum_{T_{\mathfrak{c}}\in \mathcal{T}_n\times C^n_k} 
 \prod_{i=1}^n ( d_i(T_{\mathfrak{c}}) -1
)!= \binom{2n-1}{n-k}\binom{2n+k-3}{2n-1}\times(k-2)!\;.
\end{equation*}
\end{lemma}
\begin{proof}
\begin{align}
\frac{1}{n!}
\sum_{T_{\mathfrak{c}}\in \mathcal{T}_n\times C^n_k}  \prod_{i=1}^n
(d_i(T_{ \mathfrak{c}}) -1)!&=\frac{1}{n!}\sum_{T\in\mathcal{T}_n}\prod_{i=1}^n(d_i(T)-1)!\sum_{\mathfrak{c}\in C^n_k}\prod_{i=1}^n\frac{(d_i(T)+\mathbf{1}_{i\in\mathfrak{c}}-1)!}{(d_i(T)-1)!}\nonumber\\&=\frac{1}{n!}\sum_{T\in\mathcal{T}_n}\prod_{i=1}^n(d_i(T)-1)!\sum_{\mathfrak{c}\in C^n_k}\prod_{i\in\mathfrak{c}}d_i(T)\nonumber\,.
\end{align}
Here, to count the number of trees, we use Cayley's theorem that states that:
\begin{align}\label{Cayley}
    \sum_{T\in\mathcal{T}_n}\prod_{i=1}^n(d_i(T)-1)!=(n-2)!\sum^n_{\substack{d_1,\dotsc,d_n=1
    \\ \sum_i d_i = 2(n-1)}}1\,,
\end{align}
which yields
\begin{equation*}
\frac{1}{n!}
\sum_{T_{\mathfrak{c}}\in \mathcal{T}_n\times C^n_k}  \prod_{i=1}^n
(d_i(T_{ \mathfrak{c}}) -1)!=\frac{(n-2)!}{n!}\sum^n_{\substack{d_1,\dotsc,d_n=1
    \\ \sum_i d_i = 2(n-1)}}\sum_{\mathfrak{c}\in C^n_k}\prod_{i\in\mathfrak{c}}d_i\,.
\end{equation*}
The sum over the $d_i$'s can be computed by the following trick. Let us consider the function $f$ of $n$ variables:
\begin{equation*}
f(x_1,\dotsc,x_n) = \sum_{d_1,\dotsc,d_n=1}^\infty\prod_{i=1}^n x_i^{d_i} = \prod_{i=1}^n\frac {x_i}{1-x_i} \;.
\end{equation*}
Applying the following differential operator to $f$, and evaluating it
at $(x,\dotsc,x)$ gives the expression of the sum as a Taylor coefficient:
\begin{align*}
\sum_{\substack{d_1,\dotsc,d_n=1\\\sum_i d_i = 2(n-1)}}^\infty
\prod_{i\in\mathfrak{c}} d_{i}&=[x^{2(n-1)}] \big(\prod_{i\in\mathfrak{c}}
x_{i}\frac{\partial}{\partial
  x_{i}}\big)f(x,\dotsc,x)\\ &=[x^{2(n-1)}]\frac{x^{n}}{(1-x)^{n+k}}=\binom{2n+k-3}{n-2}.
\end{align*}
With this result at hand, and using $\sum_{\mathfrak{c}\in C^n_k} 1
=n!/(n-k)!$, we finally obtain that
\begin{align*}
\frac{1}{n!}
\sum_{T_{\mathfrak{c}}\in \mathcal{T}_n\times C^n_k}  \prod_{i=1}^n
(d_i(T_{ \mathfrak{c}}) -1)!&=\frac{(n-2)!}{n!}\times\frac{n!}{(n-k)!}\binom{2n+k-3}{n-2}=\binom{2n-1}{n-k}\binom{2n+k-3}{2n-1}\times(k-2)!\,.
\end{align*}
\end{proof}

Combining this with the trigonometric identity
$2\cos^2{(x/2)}=(1+\cos x)$, we obtain the following bound on the cumulants:
\begin{multline*}
\lvert \mathfrak{K}_{\psi}^{2k}(g,\epsilon) \rvert \lesssim_k 
\frac{(k-2)!}{\sqrt{\cos{(\psi-\theta)}} } {
  \frac{1}{\cos^{k}{(\frac{\varphi+\psi}{2})}} } \\
  \times \sum_{n \geq k}
\binom{2n-1}{n-k}\binom{2n+k-3}{2n-1}{\left(\frac{\lvert g \rvert}{
    (1+\cos{(\varphi+\psi})) \sqrt{\cos{(\psi-\theta)} }
  }\right)}^{n-1}.
\end{multline*}
By Stirling's formula, $\binom{2n-1}{n-k}\lesssim_k 4^n$ and
$\binom{2n+k-3}{2n-1}(k-2)!=\frac{(2n+k-3)!}{(2n-1)!}\lesssim_k n^{k-2}$ so that we can finally bound the cumulants by:
\begin{equation*}
\lvert \mathfrak{K}_{\psi}^{2k}(g,\epsilon) \rvert \lesssim_k
\frac{1}{\sqrt{\cos{(\psi-\theta)}} } {
  \frac{1}{\cos^{k}{(\frac{\varphi+\psi}{2})}} } \sum_{n \geq k}
     {\left(\frac{4\lvert g \rvert}{ (1+\cos{(\varphi+\psi}))
         \sqrt{\cos{(\psi-\theta)} } }\right)}^{n-1} n^{k-2} \;.
\end{equation*}

We conclude on the analyticity of the cumulants thanks to two classical
theorems: the first one states that the integral of a function that
depends analytically on a paramater defines an analytic function as long as it converges;
the second one states that a uniformly convergent series of analytic
functions is analytic. 

The cumulants are expressed as a uniformly convergent series of analytic functions
both in $\Pi g$ and $\epsilon$. Since the theorems above apply in the bivariate
analytic case $\mathfrak{K}_{\psi}^{2k}(g,\epsilon)$ is analytic on the domain of
$\Sigma\times\rmH$ where the conditions \eqref{conds} hold and analytically continues $\mathfrak{K}^{2k}(g,\epsilon)$. Taking the union of these
domains for $\psi \in (\theta- \pi/ 2,\theta +\pi /2)$ yields
an analytic continuation of $\mathfrak{K}^{2k}(g,\epsilon)$ to the
subdomain $\mathfrak{C}$ of $\Sigma\times\rmH$ which concludes the proof.
\end{proof}

\begin{remark} For $g$ such that $\Pi g\in\R_-$, Borel summability in $\epsilon$ is lost since for $\varphi=\pm\pi$,
  the cardioid~\eqref{conds1} shrinks to zero when $\theta \rightarrow
  \pm \pi/2$. However, the domain of analyticity we found passes beyond the negative real axis and continues on the Riemann sheet. At the negative real axis the cumulants converge for 
  $|\Pi g|\le\frac{1}{6\sqrt{3}}$ which is of order 1. Of course, the cumulants are discontinuous here: the analytic continuations coming from above and from below the negative real axis do not coincide. 
  
  The discontinuity of the partition function and its logarithm are well understood as non perturbative instanton contributions: in zero dimensions and for $N=1$ this is detailed for instance in \cite{Aniceto2019jn}. On the contrary, the discontinuities of the cumulants have so far been less well studied.
\end{remark}

\begin{proof} It is possible to make use of $\psi$ in order to reach the negative
real axis for $g$. Indeed, assuming $\epsilon$ real positive, so that
$\theta =0$, we let $z_\psi(\varphi)=\frac 12\cos^2\del[1]{\frac{\varphi+\psi}2}\sqrt{\cos\psi}\,e^{i\varphi}$
be a point on the boundary of the
cardioid $\{|g|<\frac
12\cos^2(\frac{\varphi+\psi}2)\sqrt{\cos\psi}\}$ The maximal value of $|z_{\psi}(\pm\pi)|$ is attained for 
$\psi_0=2\arcsin\big(\frac 1{\sqrt 3}\big)$ and is $\frac 1{6\sqrt 3}$. 
\end{proof}

\section{Borel summability of the cumulants in $\mathbf{1/N}$}\label{sec5}
This last section is devoted to:

\begin{proof*}{Proof of Lemma~\ref{lemmaboundrest}}
The Borel summability of the cumulants stems from the
analyticity in a Sokal disk as stated in Theorem~\ref{THM1} 
and an estimation of the Taylor remainder. 
As we aim to obtain Borel summability in $\epsilon$ uniformly in $g$, we need to show that at large 
$q$ the remainder of order $q$ is bounded from above by 
$C\,K^q\,\lvert \epsilon \rvert^q \, q!$ with $C$ and $K$
\textit{independent of} $ g$. 

In order to compute the Taylor reminder of order $q$ we start from the expansion of
$\mathfrak{K}_{\psi}^{2k}(g,\epsilon)$ in eq.~\eqref{cumul}. We fix some $\psi \in (\theta-\pi/
2,\theta + \pi /2)$. Then, for all $k\geq 1$ and
$(g,\epsilon)\in \mathfrak{C}$ such that $\lvert \varphi+ \psi
\rvert < \pi $, the cumulants read (recall the \eqref{T_c} convention):
\begin{multline*}
\mathfrak{K}_{\psi}^{2k}(g,\epsilon) = 2^{k-1} \sum_{n \geq k}
\frac{1}{n!} {\left(\frac{- \Pi g}{2}\right)}^{n-1}
\sum_{T_{\mathfrak{c}}\in \mathcal{T}_n\times C^n_k} \int du_{T} \\
\times \bigg[
  e^{t\frac{\epsilon e^{-\imath \psi}}{2} \langle \partial,\partial
    \rangle_{W^{T}(u)} } \prod_{i=1}^n \bigg\{ (d_i(T_{\mathfrak{c}})
  -1)! \, \liftR^{d_i(T_{\mathfrak{c}})}(\sigma^{(i)} e^{\imath \frac \psi2}, g)
  \bigg\} \bigg]_{\sigma^{(i)}=0,t=1} \;, \nonumber
\end{multline*}
and the Taylor remainder of order $q$ of
$\mathfrak{K}_{\psi}^{2k}(g,\epsilon)$, denoted by
$R^{2k}_{q,\psi}(g,\epsilon)$ writes:
\begin{multline}
R^{2k}_{q,\psi}(g,\epsilon) = \int_0^1  ds \frac{(1-s)^{q-1}}{(q-1)!} 2^{k-1}
\sum_{n \geq k} \frac{1}{n!} {\left(\frac{-\Pi g}{2}\right)}^{n-1}
\sum_{T_{\mathfrak{c}}\in \mathcal{T}_n\times C^n_k} \int du_{T} \int d\mu_{s \epsilon
  e^{-\imath \psi} W^{T}(u)}(\sigma)\\
  \times {\left(\frac{\epsilon
    e^{-\imath \psi}}{2}\right)}^{q}  {\left( \langle
  \partial,\partial \rangle_{W^{T}(u)}\right)}^{q}\bigg[
  \prod_{i=1}^n \bigg\{ (d_i(T_{\mathfrak{c}}) -1)! \,
  \liftR^{d_i(T_{\mathfrak{c}})}(\sigma^{(i)} e^{\imath \frac \psi2}, g)
  \bigg\}\bigg].\label{restterm}
\end{multline}

We would like to reexpress the remainder as a sum over
some graphs. Since $2q$
derivatives are going to act on each term of the sum over the ciliated trees, and since they can act on each of the $n$
vertices, to the amplitude of a ciliated tree $T_{\mathfrak{c}}$ are
now going to correspond $n^{2q}$ amplitudes indexed by \textit{marks}
$\mathfrak{m}$ in $D^n_{2q}=\{1,...,n\}^{2q}$ corresponding to
the ordered sequence of vertices on which the derivatives are
acting (that is to say that for all $j\in\{1,...,2q\}$, the vertex $\mathfrak{m}_j$ is the vertex on which acted the $j$-th
derivative in eq.~\eqref{restterm}). This allows us to index the sum \eqref{restterm} by \textit{decorated trees}, for which we adopt the next convention:
  \begin{equation}
\text{triples $(T,\mathfrak{c},\mathfrak{m})$
made of a tree $T\in\mathcal T_n$, \emph{cilia} $\mathfrak{c}\in C^n_k$ and \emph{marks} $\mathfrak{m}\in D^n_{2q}$ are denoted by $T_{\mathfrak{c},\mathfrak{m}}$}.\tag{$\star\!$ $\star$}  \label{T_cm} 
    \end{equation}
    
For all $i\in\{1,...,n\}$, we also denote by $d_i(T_{\mathfrak{c}, \mathfrak{m}})=m(i)+d_i(T_{\mathfrak{c}})=m(i)+c(i)+d_i(T)$
the coordination degree of the vertex $i$ in the decorated tree
$T_{\mathfrak{c},\mathfrak{m}}$, with $m(i)= \modulus{\{j \in
\{1,...,2q\}\mid \mathfrak{m}_j=i\}}$ the number of marks of
$i$ and $c(i)=\textbf{1}_{i\in\mathfrak{c}}$ the number of cilia of $i$, which is 0 or 1. With this notation, the rest term
rewrites (recall that with the convention \eqref{T_cm}, $T$ is
the tree $T_{\mathfrak c,\mathfrak m}$ without cilia and marks):
\begin{multline*}
R^{2k}_{q,\psi}(g,\epsilon) = 2^{k}(-\epsilon
)^{q}\int_0^1 ds \frac{(1-s)^{q-1}}{(q-1)!}  \sum_{n \geq k}
\frac{1}{2^nn!} {\left(-\Pi g\right)}^{n-1+q}\sum_{T_{\mathfrak{c},\mathfrak{m}}\in \mathcal{T}_n\times C^n_k\times D^n_{2q}}  \\
  \times \int du_{T}\int
d\mu_{s\epsilon e^{-\imath \psi} W^{T}(u)}(\sigma) \prod_{\ell=1}^{q}
W^{T}_{\mathfrak{m}_{2\ell-1}\mathfrak{m}_{2\ell}}(u) \prod_{i=1}^n \big\{ (d_i(T_{\mathfrak{c},\mathfrak{m}}) -1 )! \,
\liftR^{d_i(T_{\mathfrak{c},\mathfrak{m}})}(\sigma^{(i)} e^{\imath \frac \psi2}, g) \big\}.
\end{multline*}

The remainder can now be bounded using the same arguments as in 
Section~\ref{sec4}, but taking into account the
combinatorics of the new $2q$ derivatives that can act on a
ciliated tree $T_{\mathfrak{c}}$.  We have the following lemma:

\begin{lemma} For all $n\geq k$, $q\geq 0$, the sum over $k$-ciliated (i.e. $\mathfrak c \in C_k^n$) and $2q$-marked (i.e. $\mathfrak m \in D_{2q}^n$) trees with $n$-vertices verifies
\begin{align}\label{combi2}
\frac{1}{n!} \sum_{T_{\mathfrak{c},\mathfrak{m}}\in \mathcal{T}_n\times C^n_k\times D^n_{2q}} \prod_{i=1}^n
(d_i(T_{\mathfrak{c},\mathfrak{m}}) - 1)! = \binom{2n-1}{n-k}
\binom{2n+2q+k-3}{2n-1} \times (2q+k-2)! \; .
\end{align}
\end{lemma}
In particular, for $q=0$ we recover Lemma~\ref{lemmacombina1}.
\begin{proof}
Injecting Cayley's formula~\eqref{Cayley} in eq.~\eqref{combi2} yields
\begin{align}
  \frac{1}{n!}\sum_{T_{\mathfrak{c},\mathfrak{m}}\in \mathcal{T}_n\times C^n_k\times D^n_{2q}} \prod_{i=1}^n
  (d_i(T_{\mathfrak{c},\mathfrak{m}}) - 1)!
  &= \frac{(n-2)!}{n!}
  \sum^n_{\substack{d_1,...,d_n=1
      \\ \sum_i d_i = 2n-2}}
  \sum_{\mathfrak{c}\in C^n_k}\sum_{\mathfrak{m}\in D^n_{2q}}
  \prod_{i=1}^n \frac{(d_i+c(i)+m(i) -
    1)!}{(d_i - 1)!}\nonumber  \\ 
     &= \frac{(n-2)!}{n!}
  \sum^n_{\substack{d_1,...,d_n=1
      \\\sum_i d_i = 2n-2}}
  \sum_{\mathfrak{c}\in C^n_k} \prod_{i\in\mathfrak{c}}
 d_i
  \sum_{\mathfrak{m}\in D^n_{2q}} \prod_{i=1}^n
  \frac{(d_i+c(i)+m(i)-
    1)!}{(d_i+c(i) - 1)!}   \; .\nonumber
\end{align}
Then, using
\begin{align*}
     \sum_{\mathfrak{m}\in D^n_{2q}} \prod_{i=1}^n
  \frac{(d_i+c(i)+m(i) -
    1)!}{(d_i+c(i) - 1)!}&=\sum_{\substack{m(1),...,m(n) \\\sum_i m(i)=2q}} \frac{(2q)!}{\prod_{i=1}^n m(i)!}  \prod_{i=1}^n
  \frac{(d_i+c(i)+m(i) -
    1)!}{(d_i+c(i) - 1)!}\\
    &=(2q)! \sum_{\substack{m(1),...,m(n) \\\sum_i m(i)=2q}} \prod_{i=1}^n \binom{d_i+c(i)+m(i)-1}{m(i)}     \nonumber\\
    &= (2q)! [x^{2q}] \prod_{i=1}^n \frac{1}{(1-x)^{d_i+c(i)}}=(2q)! [x^{2q}]\frac{1}{(1-x)^{2n-2+k}}\nonumber\\
    &= (2q)! \binom{2n-2+k+2q-1}{2q} =\frac{(2n+2q+k-3)!}{(2n+k-3)!} \,, \nonumber
\end{align*}
$\sum_{\mathfrak{c}\in C^n_k} 1 =n!/(n-k)!$ and $\sum^n_{\substack{d_1,...,d_n=1 \\ \sum_i d_i
    = 2n-2}} \prod_{i\in\mathfrak{c}}d_i = \binom{2n+k-3}{n+k-1} $, we get: 
\begin{align*}
\frac{1}{n!} \sum_{T_{\mathfrak{c},\mathfrak{m}}\in \mathcal{T}_n\times C^n_k\times D^n_{2q}} \prod_{i=1}^n
  (d_i(T_{\mathfrak{c},\mathfrak{m}}) - 1)! &=\frac{(n-2)!}{n!} \times\frac{n!}{(n-k)!}
  \binom{2n+k-3}{n+k-1} \frac{(2n+2q+k-3)!}{(2n+k-3)!}\\&=  \binom{2n-1}{n-k}
\binom{2n+2q+k-3}{2n-1} \times (2q+k-2)! \;.
\end{align*}
\end{proof}

Thanks to this lemma, we can now find an upper bound on the rest
term. All the entries of the $W^T(u)$ matrices are bounded by one, and eq.~\eqref{condsa1} implies that $\modulus{g}^q$ is smaller than $1/2^q$. We also use Lemma~\ref{complexgaussianbound} to bound the
integration over $\sigma$, and we trivially bound all the integrals
over $s$ and the $u_{T}$'s by one leading to:
\begin{multline*}
\lvert R^{2k}_{q,\psi}(g,\epsilon) \rvert \lesssim_k \,\left(\frac{\vert \epsilon \rvert}{2}\right)^{q} \frac{(2q+k-2)!}{(q-1)!} \sum_{n \geq k} \binom{2n-1}{n-k} \binom{2n+2q+k-3}{2n-1} {\left(\frac{\lvert  g \rvert}{2 }\right)}^{n-1}  \\
\times{\left(\frac{1}{\sqrt{\cos{(\psi-\theta)}    }}\right)}^{\! n}
{\left(\frac{2}{ 1+\cos{(\varphi+\psi)}}\right)}^{\! n+q+\frac{k}{2}-1}.
\end{multline*}

Now, let us choose some small $\alpha>0$, and take $(g,\epsilon)\in\mathfrak{C}_{\alpha}$, that is to say such that the
inequalities~\eqref{condsa} are satisfied. 
Note that in this domain, $g$ and $\epsilon$ satisfy
tighter bounds, which are gathered in the following lemma:
\begin{lemma}
For small $\alpha >0$, and for all $(g,\epsilon) \in \mathfrak{C}$ such that~\eqref{condsa} hold, we have:
\begin{subequations}
\begin{align}
  \alpha^2 &\leq \frac{ 1+\cos{(\varphi+\psi)}}{2} \leq 1
  \\ \sqrt{\alpha} &\leq \sqrt{\cos{(\psi-\theta)}} \leq 1 \\ \label{ineqsqrtx}
     \frac{\alpha}{2}&\leq 1 - \sqrt{\gamma }\leq 1 \qquad  \qquad \text{with } \gamma =\frac{4\lvert  g \rvert }{ {(1+\cos{(\varphi+\psi)})\sqrt{\cos{(\psi-\theta)}} }} \; .
\end{align}
\end{subequations}
\end{lemma}
\begin{proof*}{Proof.}
Since $0\leq \lvert \varphi+\psi \rvert \leq \pi(1 - \alpha)$,
$\cos^2{\frac{\pi(1-\alpha)}{2}} \leq \cos^2{\frac{\lvert \varphi+\psi
    \rvert}{2}} \leq 1$, and
$\cos^2{\frac{\pi(1-\alpha)}{2}}=\sin^2{\frac{\pi\alpha}{2}} \geq
\alpha^2$ for small $\alpha >0$. Similarly, since $0 \leq \lvert
\psi-\theta \rvert \leq \frac{\pi}{2} (1- \alpha)$,
$\cos{(\frac{\pi}{2}(1-\alpha))} \leq \cos{\lvert \psi-\theta \rvert}
\leq 1$ so that $\sqrt{\cos{(\frac{\pi}{2}(1-\alpha))}} \leq
\sqrt{\cos{\lvert \psi-\theta \rvert}} \leq 1$ and
$\sqrt{\cos{(\frac{\pi}{2}(1-\alpha))}}=\sqrt{\sin{(\frac{\pi\alpha}{2}})}
\geq \sqrt{\alpha}$ for small $\alpha >0$. Finally, since $0 \leq \gamma
\leq 1- \alpha$, $0 \leq \sqrt{\gamma}\leq \sqrt{1 - \alpha}$ and $ \sqrt{1
  - \alpha} \leq 1 - \frac{\alpha}{2}$ for small $\alpha > 0$ so that
$\frac{\alpha}{2} \leq 1-\sqrt{\gamma}\leq 1$.
\end{proof*} 

Combining this with the bounds in eqs.~\eqref{resbound} and~\eqref{GaussianComplex}, and using  $(2q+k-2)!/(q-1)! \lesssim_k
4^q q!$ and $\binom{2n-1}{n-k} \leq 4^n$ by Stirling's formula, we obtain the following upper bound on the rest term:
\begin{align*}
\lvert R^{2k}_{q,\psi}(g,\epsilon) \rvert \lesssim_k &{} {\vert \epsilon
  \rvert}^{q} 2^q q! \sum_{n \geq k} \binom{2n+2q+k-3}{2n-1}
       {\left(\frac{4 \lvert g
           \rvert}{(1+\cos{(\varphi+\psi)})\sqrt{\cos{(\psi-\theta)}}
         }\right)}^{n-1}\\
        &\times{\left(\frac{2}{
           1+\cos{(\varphi+\psi)}}\right)}^{q+\frac{k}{2}}{\left(\frac{1}{\sqrt{\cos{(\psi-\theta)}}
         }\right)}\\
        \lesssim_k&{}
       (\alpha^2)^{-q-\frac{k}{2}} (\sqrt{\alpha})^{-1} {\vert
         \epsilon \rvert}^{q} 2^q q! \sum_{n \geq k}
       \binom{2n+2q+k-3}{2n-1} {\gamma}^{n-1}.
\end{align*}

Recall that $\gamma \in [0,1-\alpha \mathopen]$ in
$\mathfrak{C}_{\alpha}$. Let us denote $f(\gamma) = \sum_{n \geq k}
\binom{2n+2q+k-3}{2n-1} {\gamma}^{n-1}$ so that:
\begin{equation*}
\lvert R^{2k}_q(g,\epsilon) \rvert \lesssim_k \alpha^{-1/2-k-2q} {\vert \epsilon \rvert}^{q} 2^q q!f(\gamma).
\end{equation*}

In order to conclude, it suffices to prove that $f(\gamma)$ is
exponentially bounded in $q$ for all $\gamma \in [0,1-\alpha \mathopen]$. 
This is stated in the following lemma:
\begin{lemma}
At large $q$, for all $\gamma \in [0,1-\alpha\mathopen]$, we have that
 \begin{equation*}
   f(\gamma) \leq  \frac{4q}{{(1-\sqrt{\gamma})}^{2q+k}}.
\end{equation*}
\end{lemma}

\begin{proof*}{Proof.}
We note that $\binom{2n+2q+k-3}{2n-1} = \frac{2q+k-1}{2n-1}\binom{2n+2q+k-3}{2n-2}$ and that for any $k$ and $n$, for $q$ large enough $\frac{2q+k-1}{2n-1} \leq 4q$ so that $\binom{2n+2q+k-3}{2n-1} \leq 4q \binom{2n+2q+k-3}{2n-2}$. This implies that for all $k \geq 1$:
\begin{align}
f(\gamma) &= \sum_{n \geq k} \binom{2n+2q+k-3}{2n-1} {\gamma}^{n-1} \leq 4q
\sum_{n \geq k} \binom{2n+2q+k-3}{2n-2} {\gamma}^{n-1} \nonumber\\ &\leq 4q \sum_{n
  \geq k} \binom{2n+2q+k-3}{2n-2} {(\sqrt{\gamma})}^{2n-2}  \leq 4q \sum_{n \geq 2k-2} \binom{n+2q+k-1}{n}
            {(\sqrt{\gamma})}^{n} \nonumber \\ &\leq 4q \sum_{n \geq 0}
            \binom{n+2q+k-1}{n} {(\sqrt{\gamma})}^{n} \;. \nonumber
\end{align}
In the second line we bound the even part of the series by the total series, using the positivity of the odd part. Observing that $\sum_{n \geq 0} \binom{n+2q+k-1}{n}
\sqrt{\gamma}^n=\frac{1}{{(1-\sqrt{\gamma})}^{2q+k}}$ we are done.
\end{proof*}

Combining this lemma with eq.~\eqref{ineqsqrtx}, we can bound $f$ uniformly as
$f(\gamma)
\lesssim_k \frac{4q}{\alpha^{2q+k}}$ for all $\gamma \in
        [0,1-\alpha\mathopen]$ and denoting $C_\alpha =
        \alpha^{-2k-1/2}$ and $K_\alpha= 6\alpha^{-4}$ at $q$ large enough, for all $k \geq 1$:
\begin{align}\nonumber
\lvert R^{2k}_{q,\psi}(g,\epsilon) \rvert \lesssim_k C_\alpha K_\alpha^q {\lvert \epsilon \rvert}^{q} q! 
\; ,
\end{align}
with $C_\alpha$ and $K_\alpha$ independent of $g$ for 
$g\in\mathfrak{C}_{\alpha}$, which concludes the proof of
lemma~\eqref{lemmaboundrest}.\end{proof*}

\begin{remark}
Note that, as $K_\alpha \sim 
O(1) \alpha^{-4}$ and $C_\alpha \sim
{\alpha^{-2k-1/2}}$ our bounds deteriorate for $\alpha\to 0$ that is when 
we take a subdomain closer and closer to the full $\mathfrak{C} $. 
\end{remark}

\section*{Conclusion}
The Loop Vertex Expansion made possible to extract the logarithm
of the partition function, obtain the maximal analyticity domain (Thm. \ref{THM1}), and the domain of $1/N$-Borel summability (Thm. \ref{THM2}) of the cumulants of the quartic
$\grpO(N)$-vector model.  \par 
The next step would be to adapt our analysis to the more involved case
of a (Euclidean) quantum field theory. The fist case of interest is
the two dimensional quartic $\grpO(N)$-vector model, whose
renormalisation is limited to the Wick ordering. Two dimensional
quantum field theory was studied with a modification of the LVE known
as the Multiscale LVE (MLVE) in \cite{Rivasseau2014aa}, where the
Borel summability of free energy in the coupling constant is
established. This study should be generalized to an $\grpO(N)$-vector
model. However, the adaptation of the (M)LVE beyond dimension two seems out of reach. 

\newpage
\appendix
\appendixpage
\section{The Nevalinna-Sokal theorem}
A formal power series $A(z)=\sum_{k=0}^\infty a_k z^k$ such that $B(t) = \sum_{k=0}^\infty  \textfrac{a_k}{k!} t^k $ is absolutely convergent in some disk centered at zero an admits an analytic continuation along the real axis such that $|B(t)| < K e^{-t/R}$ for some $K,R\in \R_+$ is called a \emph{Borel summable series}. The function  $\sum_{k=0}^{\infty} a_k/k!t^k$ is called the \emph{Borel transform} of $A(z)$ and the \emph{Borel sum} of $A(z)$ is the Laplace transform of its Borel transform:
\begin{align} \nonumber
f(z) =  \frac{1}{z}  \int_0^{\infty} dt \; e^{-t/z} B(t) \;.    
\end{align}
The Borel sum of a series, if it exists, is unique.

A function $f:\C\rightarrow\C$ which is analytic in a disk tangent to the imaginary axis in $0$ and
has an asymptotic series in $0$ (which can have zero radius of convergence) such that the Taylor rest term of order $q$ in $0$ of $f$ grows no faster than $q!$ is called a \emph{Borel summable function} \cite{Sokal1980aa}.

These two notions are intimately related: the Borel sums of Borel summable series are Borel summable functions (this is straightforward to prove). The asymptotic series of Borel summable functions are Borel summable series \cite{Sokal1980aa}. 
We present here a slightly modified version \cite{Rivasseau2007aa} of the (optimal) Nevanlinna-Sokal theorem on Borel summability which introduces the notion of uniform Borel summability with respect to some parameter.

\begin{theorem}[Nevanlinna-Sokal]\label{thm:Sokal}
Let $f$  be a function $f:\C^2\rightarrow\C$, $(z,w)\mapsto f(z,w)$ and let $(f_k(w))_{k\geq0}$ be the coefficients of the asymptotic series of $f$ in $z=0$, at fixed $w\in\C$. If $f$ is analytic in its first variable $z$ in a domain ${\rm Disk}_R  = \{z\in\C \mid\Re {(1/z)}>1/R  \}$ with $R>0$ independent of $w$ and there is a domain $\cal D \subset \C$ such that for
all $(z,w)\in {\rm Disk}_R \times {\cal D} $, there exist some
constants $C,K>0$ \textit{independent of $w$} for which the following bound on the rest term of $f$ holds for $q$ large enough:
\begin{equation*}
    \left\vert f(z,w) - \sum_{k=0}^q f_k(w)z^k\right\vert \leq C K^q \modulus{z}^q q!,
\end{equation*}
then $f$ is called a Borel summable function in $z$ uniformly in $w$ in ${\rm Disk}_R \times {\cal D}$.

Under these conditions, for all $w\in\cal D$, the Borel transform in $z$ of the asymptotic series of $f$, \[\mathcal{B}:(t,w)\mapsto \sum_{k=0}^{+\infty}
\frac{f_k(w)}{k!}t^k, \nonumber\] 
is analytic in a disk of radius $K^{-1}$ in $t$ and can be analytically continued to the strip $ \{t\in \C\mid
\modulus{ \Im t } <K^{-1}\}$ and in this strip obeys the exponential bound $|\mathcal B(t)|<e^{t/R}$. Moreover, for all $(z,w)\in {\rm Disk}_R \times {\cal D} $ we can reconstruct the function $f(z,w)$ by:
\begin{align}\nonumber
f(z,w)=\frac{1}{z}\int_{\R_+} dt \; e^{-t/z} \mathcal{B}(t,w).
\end{align}
\end{theorem}
\begin{remark}\label{remarksokaldisk}
For $R>0$, the domain $\{z\in\C
\mid\Re {(1/z)}>1/R \}$ is a disk of diameter $R$
tangent to the imaginary axis at the origin. We call Sokal disk of diameter $R$ such a disk, see Figure \ref{fig:Sokal}. 
\end{remark}

\begin{figure}[ht]\centering
\includegraphics[width=7cm]{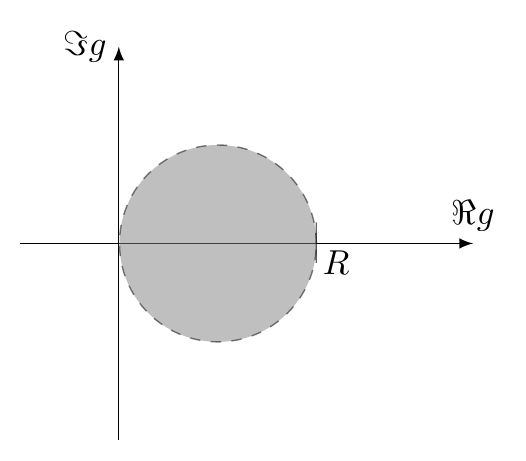}
\caption{A Sokal disk $ {\rm Disk}_R = \{ g\in\C \mid \Re (1/g) > 1/R \} $ \label{fig:Sokal}}
\end{figure}

\newpage
\section{Finite dimensional Gaussian measures}
\subsection{Gaussian expectations}\label{gaussexp}
Let $n$ a strictly positive integer. We are interested in the centered
Gaussian distributions on $\mathbb{R}^n$, that is to say the centered
probability distributions on $\mathbb{R}^n$ such that their cumulants
of order higher than or equal to three are zero. \\

\noindent
\textbf{Case $\mathbf{n=1}$.} Let us
first consider the one dimensional case. In this case, for $\sigma >0$, the Normal 
distribution of variance $\sigma^2$, $\mathcal{N}(0,\sigma^2)$ is Gaussian and
its density with respect to the Lebesgue measure on
$\mathbb{R}$ is
$(\sqrt{2\pi}\sigma)^{-1}e^{-\frac{x^2}{2\sigma^2} } $. But this is
not the only Gaussian distribution on $\mathbb{R}$: the Dirac
distribution $\delta$ whose expectation $\mathbb{E}_{0}$ is defined by
$\mathbb{E}_{0}[F(x)]=F(0)$ for all functions $F : \mathbb{R}
\rightarrow \mathbb{C}$ is also Gaussian with variance 0. There is no other 
Gaussian distributions on $\mathbb{R}$, so that the
Gaussian distributions on $\mathbb{R}$ are determined by
their variances and share the following property.
\begin{definition}[Gaussian distributions in dimension one]
For all $\varepsilon \in \mathbb{R}_{+}$, there exists a unique
centered Gaussian distribution of variance $\varepsilon$. Let us
denote it $\mu_{\varepsilon}$, and $\mathbb{E}_{\varepsilon}$ the
expectation with respect to $\mu_{\varepsilon}$. For $F : \mathbb{R}
\rightarrow \mathbb{C}\in L^1(\R,\mu_{\varepsilon})$,
$\mathbb{E}_{\varepsilon}$ is defined by the following identity:
\begin{align}\nonumber
    \mathbb{E}_{\varepsilon}[F(x)] = [e^{\frac{\varepsilon}{2} \partial_x^2}F(x)]_{x=0} \;.
\end{align}
Furthermore, if $\varepsilon \neq 0$, $\mu_{\varepsilon} =
\mathcal{N}(0,\varepsilon)$ and $\mu_{\varepsilon} = \delta$
otherwise.
\end{definition}
The previous definition immediately implies that for all $F :
\mathbb{R} \rightarrow \mathbb{C}$, $\mathbb{E}_{\varepsilon}[F(x)]
\xrightarrow[\varepsilon \rightarrow 0]{} \mathbb{E}_{0}[F(x)]$,
which means that $\mathcal{N}(0,\varepsilon) \xrightarrow[\varepsilon \rightarrow 0]{\text{in law}} \delta$. \\
\textbf{Case $\mathbf{n \geq 2}$.} The Gaussian distribution on $\mathbb{R}^n$ is
a straightforward generalization of that on $\mathbb{R}$, as stated in
the following definition:
\begin{definition}[Gaussian distributions in dimension $n$]
Let $C$ in $M_n(\mathbb{R})$ a symmetric positive matrix \textit{not
  necessarily invertible}. There exists a unique centered Gaussian
distribution of covariance $C$. Let us denote it $\mu_{C}$, and
$\mathbb{E}_{C}$ the expectation with respect to $\mu_{C}$. For $F :
\mathbb{R}^n \rightarrow \mathbb{C} \in L^1(\R^n,\mu_{C})$,
$\mathbb{E}_{C}$ is defined by the following identity:
\begin{align}\nonumber
\mathbb{E}_{C}[F(X)] = [e^{\frac{1}{2}\langle \partial,\partial \rangle_C}F(X)]_{X=0} \;.
\end{align}
Furthermore, if $C$ is in $GL_n(\mathbb{R})$, $\mu_C=\mathcal{N}(0,C)$
where $\mathcal{N}(0,C)$ is the Normal distribution of covariance $C$
that has density $\frac{1}{\sqrt{(2\pi)^n \det{C}}}e^{-\frac{1}{2}
  \langle X,X \rangle_{C^{-1}}} $ with respect to the Lebesgue measure
on $\mathbb{R}^n$. If $C$ is not invertible, then
$C_{\varepsilon} = C + \varepsilon P$, with $P$ the projector on the
kernel\footnote{Suppose $C$ has rank $k<n$. Then there exists
$\lambda_1,\dotsc,\lambda_k$ non-negative and $O \in \grpO_n(\R)$ such that
$C = O^T
\mathrm{diag}(\lambda_1,\dotsc,\lambda_k,\underbrace{0,\dotsc,0}_{n-k \text{
    times}})O$. Then, $P = O^T \mathrm{diag}(\underbrace{0,\dotsc,0}_{k
  \text{ times}},\underbrace{1,\dotsc,1}_{n-k \text{ times}})O$,
$C_\varepsilon = C+\varepsilon P \in \mathrm{GL}(\R^n)$ for all $\varepsilon
>0$, and $C_\varepsilon \rightarrow C$ when $\varepsilon \rightarrow
0$.  }  of $C$ is invertible and:
\begin{align*}
    \mathcal{N}(0,C_{\varepsilon}) \xrightarrow[\varepsilon \rightarrow 0]{\text{in law}} \mu_C  \;,
\end{align*}
which implies that:
\begin{equation*}
    \mathbb{E}_{C}[F(X)] = \lim_{\varepsilon \rightarrow 0}
    \mathbb{E}_{\mathcal{N}(0,C_{\varepsilon})}[F(X)] =
    \lim_{\varepsilon \rightarrow 0} \int_{\mathbb{R}^n}
    \frac{1}{\sqrt{(2\pi)^n
        \det{C_{\varepsilon}}}}e^{-\frac{1}{2}\langle X,X
      \rangle_{C_\varepsilon^{-1}}} F(X) d^n X \;.
\end{equation*}
\end{definition}
Once again, there are two ways of thinking of $\mu_C$ if $C$ is not
invertible: either we see it as a differential operator or as the limit 
in law of a sequence of Normal distributions. Both are
useful: the former makes the interpolation
between different covariances more transparent, while the latter makes bounding the 
expectations easier.

\subsection{Complex Gaussian integration}\label{app:compbound}

\begin{definition}[Complex Gaussian expectation] Let $n$ be a positive integer, $z=\modulus{z}e^{\imath\alpha}\in\{\Re z>0\}$, $C \in
M_n(\R)$ symmetric positive semi-definite, and $F \in
L^1(\R^n,\mu_{{\textfrac{\lvert z\rvert^2 C}{\Re z}}})$ a $\C$-valued function. We call complex Gaussian integration of covariance $zC$ the quantity denoted
  $\mathbb{E}_{zC}$ and defined by:
\begin{equation*}
\mathbb{E}_{zC}[F(X)] =[e^{\frac{z}{2}\langle \partial,\partial
    \rangle_C}F(X)]_{X=0}= \lim_{\varepsilon \rightarrow 0}
\mathbb{E}_{\frac{\lvert
    z\rvert^2}{\Re z}C_\varepsilon}[\frac{1}{\sqrt{(\cos \alpha\; e^{\imath\alpha})^n}} e^{\frac{\imath\Im z}{2\lvert
      z\rvert^2}\langle X,X \rangle_{C_\varepsilon^{-1}}}F(X)]
\end{equation*}
with $C_\varepsilon$ a sequence such that
$\mathcal{N}(0,C_\varepsilon) \rightarrow \mu_C$ (if $C$ is
invertible, take $C_\varepsilon=C$ constant).
\end{definition}
\begin{lemma}[Complex Gaussian bound]\label{complexgaussianbound} With the same notations,
  if $z = |z|e^{\imath \alpha}$ with $\alpha \in
  (-\frac\pi2,\frac\pi2)$, we have that:
\begin{equation}\nonumber
 \left| \mathbb{E}_{zC}[F(X)]\right| \le \frac{1}{\cos^{n/2}\alpha }
 \; \sup_{X \in \mathbb{R}^n} | F(X ) |\;.
\end{equation}
\end{lemma}
\begin{proof} Since for a convergent sequence the modulus and the
  limit commute,
  \begin{align*}
    |\mathbb{E}_{zC}[F(X)] |&=| \lim_{\varepsilon \rightarrow 0}
                              \mathbb{E}_{\frac{\lvert
                              z\rvert^2}{\Re z}C_\varepsilon}[\frac{1}{\sqrt{(\cos \alpha e^{\imath\alpha})^n}} e^{\frac{\imath\Im z}{2\lvert
                              z\rvert^2}\langle X,X \rangle_{C_\varepsilon^{-1}}}F(X)]|\\
                            &= \lim_{\varepsilon \rightarrow 0}|\mathbb{E}_{\frac{\lvert
                              z\rvert^2}{\Re z}C_\varepsilon}[\frac{1}{\sqrt{(\cos \alpha e^{\imath\alpha})^n}} e^{\frac{\imath\Im z}{2\lvert
                              z\rvert^2}\langle X,X \rangle_{C_\varepsilon^{-1}}}F(X)]|\\
                            &\leq  \sup_{X \in \mathbb{R}^n} \big| \frac{1}{\sqrt{(\cos \alpha e^{\imath\alpha})^n}} e^{\frac{\imath\Im z}{2\lvert
                              z\rvert^2}\langle X,X \rangle_{C_\varepsilon^{-1}}}F(X) \big|=\frac{1}{\cos^{n/2}\alpha }
                              \; \sup_{X \in \mathbb{R}^n} | F(X ) |\;.
  \end{align*}
\end{proof}

\subsection{The copies trick}\label{replicatr}
\begin{lemma}[The copies trick]\label{replica}
Let $n$ be a positive integer, $z\in\{\Re z>0\}$ and 
 $F \in L^n(\R,\mu_{\textfrac{\lvert z\rvert^2}{\Re z} 
})$ a $\C$-valued function. Then $F^{\otimes n }:\R^n\to\C$, $(X_i)_{1\leq i\leq n} \mapsto
\prod_{1=n}^n F(X_i)$ is in $L^1(\R^n,\mu_{{\textfrac{\lvert
      z\rvert^2\mathbbm{1}_n}{\Re z}} })$ and furthermore we
have:
\begin{align} \nonumber
 \mathbb{E}_{z }[ F^n(x)]  = \mathbb{E}_{z  \mathbbm{1}_n}[F^{\otimes n}(X)]
\end{align}
\end{lemma}
\begin{proof}
For simplicity we take $z=1$. We denote by $u$ the vector of
$\mathbb{R}^n$ with all entries $1/\sqrt{n}$, such that $\mathbbm{1}_n
=n u \otimes u$. Let $v_2,...,v_n \in \R^n$ such that
$(u,v_2,...,v_n)$ is an orthonormal basis of $\R^n$. We aim to
understand the action of $\mu_{\mathbbm{1}_n}$ on a test function $G :
\R^n \rightarrow \C$. For $\varepsilon>0$, let us define
$C_\varepsilon$ by $C_\varepsilon = \mathbbm{1}_n + \varepsilon
\sum_{i=2}^n v_i \otimes v_i$ so that $C_\varepsilon \rightarrow
\mathbbm{1}_n $ as $\varepsilon \rightarrow 0$. Then,
\begin{align}
    \mathbb{E}_{\mathbbm{1}_n}[G(X)] &= \lim_{\varepsilon \rightarrow
      0} \int_{\mathbb{R}^n} \frac{1}{\sqrt{(2\pi)^n
        \det{C_{\varepsilon}}}}e^{-\frac{1}{2}\langle X,X
      \rangle_{C_\varepsilon^{-1}}} G(X) d^n X\nonumber
    \\ &= \lim_{\varepsilon
      \rightarrow 0} \int_{\mathbb{R}^n} \frac{1}{\sqrt{(2\pi)^n
        n\varepsilon^{n-1}}}e^{-\frac{1}{2}(n^{-1}y_1^2+\varepsilon^{-1}
      \sum_{i=2}^n y_i^2)} G(y_1 u +\sum_{i=2}^n y_iv_i) \,dy_1
    \prod_{i=1}^n dy_i \nonumber
    \\ &= \lim_{\varepsilon \rightarrow
      0} \int_{\mathbb{R}^n} \frac{1}{\sqrt{(2\pi)^n
        \varepsilon^{n-1}}}e^{-\frac{1}{2}(y_1^2+\varepsilon^{-1}
      \sum_{i=2}^n y_i^2)} G(y_1 \sqrt{n}u +\sum_{i=2}^n y_iv_i)
    \,dy_1 \prod_{i=1}^n dy_i \nonumber
    \\ &=\lim_{\varepsilon
      \rightarrow 0} \mathbb{E}_{1\otimes
      \varepsilon^{\otimes(n-1)}}[G(y_1\sqrt{n}u+\sum_{i=2}^n
      y_iv_i)]\nonumber
    \\ &= \mathbb{E}_{1\otimes
      0^{\otimes(n-1)}}[G(y_1\sqrt{n}u+\sum_{i=2}^n y_iv_i)]\nonumber
    \nonumber \\ &= \mathbb{E}_1[G(y_1,...,y_1)]\nonumber
\end{align}
To go from the first to the second line, we perform a change of
variable from $X$ to $Y=y_1u+\sum_{i=2}^n y_iv_i$ whose Jacobian
determinant is 1, and to go from the second line to the third line, we
perform the change of variable $y_1$ becomes $\sqrt{n}y_1$. Line four
is a simple rewriting of line three with $y_1 \sim \mu_1$ and $y_i
\sim \mu_{\varepsilon}$ for all $2\leq i\leq n$, and going from line
four to line five uses the convergence in law of the Normal
distribution to the Dirac distribution, while line six is simply the
expectation of the Dirac measure. \\

Applying the previous equality to $G=F^{\otimes n}$ and substituting
$\mathbbm{1}_n$ with $z \mathbbm{1}_n$ concludes the proof.
\end{proof}

\newpage
\printbibliography

\end{document}